\documentclass[11pt,letterpaper]{article}

\usepackage{amsfonts,amsmath,amssymb}
\usepackage{amsthm}
\usepackage[T1]{fontenc}

\usepackage[dvips]{graphicx}
\usepackage[usenames,dvipsnames]{xcolor}

\usepackage{CJKutf8}

\usepackage[round]{natbib}
\usepackage{multirow}
\usepackage{caption}
\usepackage{float}

\newcommand{\ER}{Erd\H{o}s-R\'{e}nyi }

\newcommand{\phit}[1]{{\tilde{\phi}}}
\newcommand{\at}[1]{\tilde{a}}

\newtheorem*{definition}{Definition}
\newtheorem{assumption}{Assumption}
\newtheorem{proposition}{Proposition}

\newcommand{\vecm}{{\bf m}}
\newcommand{\vecw}{{\boldsymbol w}}

\title{Unstable diffusion in social networks\footnote{
Kobayashi acknowledges financial support from JSPS KAKENHI\ 19H01506 and 20H05633. Onaga acknowledges financial support from JSPS KAKENHI 19K14618 and 19H01506. We would like to thank Takashi Shimizu, who provided valuable comments on an earlier version of this paper.}}

\date{\today}
\author{ Teruyoshi Kobayashi\thanks{kobayashi@econ.kobe-u.ac.jp} \\
Department of Economics, Center for Computational Social Science,\\ Kobe University, Kobe, Japan 
\and Yoshitaka Ogisu \\
Graduate School of Economics, Kobe University, Kobe Japan
\and Tomokatsu Onaga\\ Frontier Research Institute for Interdisciplinary Sciences, \\ Graduate School of Information Sciences, \\Tohoku University, Sendai, Japan}

\begin{document}
\begin{CJK*}{UTF8}{gbsn}

\maketitle

\begin{abstract} \setlength{\baselineskip}{15pt}
How and to what extent will new activities spread through social ties? Here, we develop a more sophisticated framework than the standard mean-field approach to describe the diffusion dynamics of multiple activities on complex networks. We show that the diffusion of multiple activities follows a saddle path and can be highly unstable. In particular, when the two activities are sufficiently substitutable, either of them would dominate the other by chance even if they are equally attractive \emph{ex ante}. When such symmetry-breaking occurs, any average-based approach cannot correctly calculate the Nash equilibrium --- the steady state of an actual diffusion process.  
{\flushleft JEL \emph{classification:} C72, D85, L14}
\end{abstract}

\section{Introduction}

Social activities, new ideas, and innovative technologies spread through networks of social ties formed by friends, colleagues, and followers in social media, such as Facebook, Twitter, and Instagram. 
Social ties (in physical and online spaces) are not only a channel through which information flows, but also a channel through which influence is propagated.
For example, an individual’s decision regarding which activities to join often depends on the fraction and/or the number of friends participating in the activities.
If the majority of friends already joined a certain activity, the activity would be more attractive to that individual than the activities not popular among friends.  The contagious aspect of peer effects through social ties has been extensively studied in models of coordination games on networks~\citep{morris2000contagion,Jackson2008book,young2011dynamics,kreindler2014rapid,jackson-zenou2015games}, in which they examine the effect of network structure on the possibility of cascades. Another framework is based on a utility-maximization problem that incorporates gains from interacting with adjacent players~\citep{ballester2006,chen2018AEJmultiple}. They prove that the equilibrium strategies are determined by the players' positions in the network characterized by their Bonacich centralities~\citep{bonacich1987power}.

In the standard $2\times 2$ coordination games with a single activity (or product, convention, opinion, whatever), the players' decisions are binary (e.g., do it or don't do it). 
However, when there are two competing activities, denoted by $A$ and $B$, the payoff matrix generally becomes $4\times 4$ and each player's option is no longer binary; there are four pure strategies in the strategy set $\{00,01,10,11\}$, where ``$00$'' denotes the strategy of not joining either activity (or the status quo), ``$01$'' (resp.~``$10$'') denotes the strategy of joining activity~$A$ only (resp.~$B$ only), and ``$11$'' denotes the \emph{bilingual option} of joining both $A$ and $B$~\citep{oyama2015bilingual,arigapudi2020bilingual}.\footnote{In \cite{oyama2015bilingual} and \cite{arigapudi2020bilingual}, the strategy ``$00$'' (i.e., do-nothing strategy) is not considered, and each player selects either or both of the two activities. In our study, the strategy ``$00$'' is the status quo for all players except the \emph{seed players} that are initially activated.} 
 The presence of multiple activities allows us to study richer dynamics compared to the binary cascade models.
 First, while there is only one type of strategic shift in the diffusion of a single activity (i.e., $00\to 01$), there are multiple patterns of strategic shifts, e.g., $01\to 10$, $01\to 11$, $10\to 11$, etc, which simply expands the set of possible paths to an equilibrium. 
 Second, we can introduce the reversibility of strategic choices; players may revert their strategies, e.g., $01\to 10\to 01$, in response to changes in the neighbors' states. The propagation process is thus no longer monotonic, where an activity may widely spread temporarily but fade away to be replaced by the other activity, while diffusion is always monotonic in the binary cascade models.\footnote{Note that the standard binary cascade models do not exhibit such a non-monotonic behavior, allowing us to prove the stability of an equilibrium and convergence of an iteration algorithm~\citep{jackson2007diffusion,kobayashi2021dynamics}.} Third, the diffusion process can be more stochastic especially when the two activities are equally attractive. In complex networks, whose structures are far from regular and symmetric, the influence of initially active players is highly heterogeneous. Thus, the popularity of an activity in the terminal state will be strongly affected by the extent to which the activity spreads in the early stage of diffusion.

 When the network structure is complex, unlike regular lattices, star graphs, and circular networks, it is notoriously difficult to calculate the Nash equilibrium in an exact manner even in binary games.
 In the literature, it is common to calculate the Nash equilibrium using a \emph{mean-field} (MF) approximation~\citep{jackson2007diffusion,lopez2008GEBdiffusion,lopez2012GEBinfluence,lelarge2012diffusion,sadler2020diffusion}, assuming that players are sufficiently homogeneous such that the probability of each player being in a given state can be well approximated by the corresponding average over all players.  However, \cite{gleeson2018message} and \cite{kobayashi2021dynamics} quantitatively verify that the MF approximation can be inaccurate, especially when the network connectivity is close to the critical points at which the size of cascades changes drastically. They show that a \emph{message-passing approach} is always more precise than the conventional MF approximation and that an iteration algorithm surely converges to a fixed point that corresponds to the simulated Nash equilibrium.

 In this paper, we study two classes of games that describe multiple-activity diffusion on complex networks: i) coordination games with multiple activities, also known as \emph{bilingual games}~\citep{goyal1997non,immorlica2007role,oyama2015bilingual}, and ii) the utility-based games on networks proposed by \cite{ballester2006} and \cite{chen2018AEJmultiple}. 
 A common property of the two classes of games is that the optimal strategic choices are given by a set of threshold rules: \emph{fractional threshold rules}~\citep{Watts2002} in coordination games, and \emph{absolute threshold rules}~\citep{Granovetter1978} in the utility-based games. 
 In each of these classes, we calculate the Nash equilibrium highly accurately by solving a system of differential equations, called the \emph{approximate master equations} (AME)~\citep{gleeson2011high,gleeson2013binary}. 
 The key benefit of analyzing  with a system of differential equations is that it gives us the dynamical path of the popularity of each activity, i.e., the extent to which each activity spreads over the network at any given point in time.
 While the AME approach is generally more accurate than MF approximations, the simplicity of the MF equations allows us to provide an analytical description of diffusion dynamics using phase diagrams. 
 
 The main results are summarized as follows. First, the system of equations given by the AME approach reveals the diffusion dynamics of multiple activities highly accurately in the sense that the calculated paths well match the simulated diffusion processes and correctly predict the Nash equilibrium. On the other hand, the MF equations, while they are simpler and more analytically tractable, replicate the simulated paths only roughly, and the predicted Nash equilibria can deviate from the simulated ones.\footnote{The inaccuracy of the MF equations is also pointed out by \cite{gleeson2013binary}.} 
 
 Second, we find that there are four distinct regimes characterized by different Nash equilibria, depending on the relative attractiveness of the activities and connectivity of the network. 
 At the boundary of these regimes, a slight change in the relative attractiveness or network connectivity may drastically shift the equilibrium, a phenomenon called \emph{phase transition}~\citep{Watts2002,GaiKapadia2010}.
 Thus, a small change in the attractiveness and/or the network structure may initiate or terminate the widespread diffusion of an activity. This suggests that the diffusion dynamics are unstable at the critical points, and the popularity of an activity is far from proportional to its intrinsic attractiveness.
 
 Third, except at those boundaries, the diffusion dynamics are stable and the equilibrium is highly predictable when the two activities are complementary or neutral. However, this is not necessarily the case when they are substitutes. Suppose that the two activities are perfectly symmetric in the sense that their attractiveness is represented by exactly the same payoff/preference parameters. In theory, we always obtain a common solution for the cascade sizes of symmetric activities, e.g., 40\% of the population adopts activity $A$ and the remaining 40\% adopts activity $B$, because there is no factor that differentiates between the two, at least on average. In fact, this is not necessarily the case in the simulated diffusion processes. The analytical solutions would indeed be correct if the topological properties of the initially active players (i.e., \emph{seed nodes}) are symmetric, but such a situation rarely occurs in complex networks. Our numerical experiments reveal that diffusion processes on complex networks generated from the same degree distribution can still reach totally different Nash equilibria when the two activities are substitutes. This suggests that equally-attractive yet substitutable activities do not necessarily gain equal popularity, and either activity may even dominate the other through the cascade of peer effects.

 The possibility of ``\emph{symmetry breaking}'' also raises an important issue for theoretical studies of network games.  With such an unstable equilibrium, any deterministic equilibrium would fail to predict the ``true'' Nash equilibrium since one of the possible equilibria will inevitably be achieved by chance, depending on the details of network structure that cannot be captured by the degree distribution. Symmetry breaking has long been recognized as a source of diversity in economic development~\citep{matsuyama2002explaining,acemoglu2017asymmetric}, financial globalization~\citep{matsuyama2004financial}, and international trade~\citep{matsuyama2013trade,chatterjee2017endogenous}. To the best of our knowledge, this is the first study that shows why almost equally attractive activities (or technologies, products, etc) can gain totally different levels of popularity. 
 
 In the theoretical analysis, we use stylized networks such as \ER random graphs~\citep{Erdos1959PublMath} and random regular graphs, with which we can obtain exact degree distributions. While this greatly facilitates the analysis, these networks are not necessarily realistic~\citep{Barabasi2016book}.
 To check how well our model would explain diffusion processes on real social networks, we also examine the goodness of fit using empirical data.  Here, we construct a social network of economists based on the acknowledgments of articles published in \emph{American Economic Review} between 2019 and 2020, where nodes represent authors and edges correspond to social ties created by giving and receiving comments.

\section{Diffusion through coordination games}
\subsection{Network structure}

For theoretical analysis, we construct synthetic networks formed by $N$ players, where $N$ is assumed to be sufficiently large. Player $i$ is connected to $k_i$ other players by undirected and unweighted edges. $k_i$ is called the \emph{degree} of player $i$ (or node $i$).
Players at the end of the edges emanating from $i$ are called the \emph{neighbors} of player $i$. 
We consider a whole ensemble of many possible networks in a given class, where a particular network structure is realized with a certain probability. That is, we do not focus on a particular single network, rather we specify the distribution of all possible network structures that would appear in a given network model. Any realized network is, therefore, an instance drawn from the ensemble uniformly at random. Examples of \ER networks~\citep{Erdos1959PublMath} and random regular graphs are shown in Fig.~\ref{fig:random_graphs}a.

\begin{figure}[tb]
    \centering
    \includegraphics[width=15.5cm]{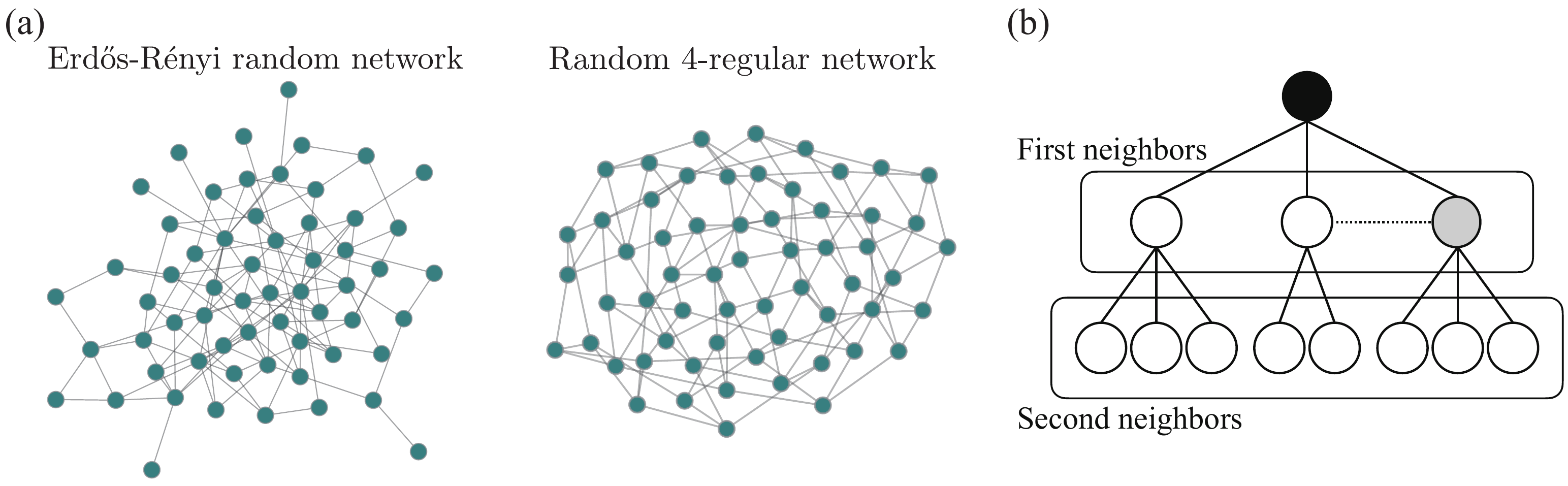}
    \caption{Structure of random networks. (a) Visualization of \ER random networks and random $z$-regular networks. $N=60$ and $z=4$ in both cases. (b) Schematic of a tree-like structure ($z=3$).}
    \label{fig:random_graphs}
\end{figure}

 In \ER networks, any two nodes are connected with probability $q$, and the model is expressed as $G(N,q)$. Let $M$ denote the total number of (undirected) edges.
In the $G(N,q)$ model, the probability that a realized network has $x$ edges is given by
\begin{align}
    \text{Prob}(M=x;G) = \binom{\binom{N}{2}}{x} q^x (1-q)^{\binom{N}{2}-x},
\end{align}
where the probability of observing each particular network with $x$ edges is $q^x (1-q)^{\binom{N}{2}-x}$.
Let $q_k$ denote the degree distribution, namely, the probability that a randomly chosen node has degree~$k$. In the \ER model, $q_k$ is given by a binomial distribution, so it is approximated by a Poisson distribution for large $N$:
\begin{align}
    q_k = \binom{N-1}{k}q^{k}(1-q)^{N-1-k}\approx \frac{z^k e^{-z}}{k!},
\end{align}
where $z=\sum_{k}q_k k$ denotes the average degree.
Practically, we need to determine the maximum value for degree $k$, denoted by $k_{\rm max}$. Throughout the analysis, we set a sufficiently large $k_{\max}$ such that it covers at least $99.9\%$ of the entire distribution, i.e., $k_{\rm max}=\min\{k':\sum_{k=0}^{k'}q_k\geq 0.999\}$.

In random $z$-regular networks, nodes are connected uniformly at random subject to the degree constraint that $k_i=z$ for every $i=1,\ldots, N$. The number of edges is thus prespecified as $M= zN/2$ since every node has degree~$z$.  While we mainly use the \ER model in the following analysis, we will also use, when necessary, random $z$-regular networks to maintain analytical tractability.\footnote{Asymptotic properties of random regular graphs are discussed in \cite{wormald1999models} and \cite{mckay2004short}.}

While we use \ER networks and random $z$-regular graphs in our analysis, the only network property required for our approach is a locally tree-like structure with which the presence of local cycles can be ignored.
For instance, \ER networks are locally tree-like if the connecting probability $q=z/(N-1)$ is sufficiently small (i.e., $z\ll N$)~\citep{newman2018book2nd}. If a network is locally tree-like, the average number of second neighbors, i.e., the nodes at distance two from a starting node, is given by $\sum_{k,\ell}q_k\tilde{q}_\ell k\ell=\sum_{k}q_kk \sum_\ell q_{\ell+1}\ell(\ell+1)/z = z\sum_{\ell} q_\ell(\ell-1)\ell/z=  z^2$, where $\tilde{q}_\ell \equiv (\ell+1)q_{\ell+1}/z$ denotes the excess degree distribution\footnote{The excess degree distribution $\tilde{q}_\ell$ is the probability that the total degree of a randomly selected neighbor is $\ell+1$, i.e., the probability that a randomly selected neighbor of a starting node has $\ell$ edges except the edge emanating from the starting node~\citep{newman2018book2nd}. $\ell$ is called the \emph{excess degree}.}, and we used the fact that the variance of degrees is given by $z$ when the degree distribution is Poissonian. This indicates that each of the $z$ first neighbors of a starting node connects to $z$ second neighbors, on average, and the first and second neighbors of the same starting node will not overlap (Fig.~\ref{fig:random_graphs}b).
To give an intuition, consider a representative situation in which a starting node has $z$ (first) neighbors and one of the neighbors has an excess degree $z$.
Since nodes are connected uniformly at random, the probability of a first neighbor being connected to at least one of the other first neighbors (e.g., dotted line in Fig.~\ref{fig:random_graphs}b) is $1-\left[1-\frac{z-1}{N-2}\right]^{z}$, where $(z-1)/(N-2)$ denotes the probability that one of the other $z-1$ first neighbors is chosen among $N-2$ nodes, excluding the starting node (black node) and the focal first neighbor (gray node). 
This suggests that the chance that a local triangle is formed will be vanishingly small if $z\ll N$, where it will be unlikely that the first neighbors of the starting node are connected with each other. 
Since this argument holds for any starting node, the neighbors of a first neighbor are also unlikely to be connected with each other, and the neighbors of a second neighbor are unlikely to be connected with each other, and so on.\footnote{See \cite{newman2018book2nd} ch.~12 for discussion of the tree-like structure.}  

A more general class of network models, called \emph{configuration models}, in which the degree distribution is prespecified while nodes are connected at random subject to the degree constraint, also exhibit a locally tree-like property~\citep{molloy1995critical,newman2018book2nd}. Clearly, \ER networks and $z$-regular networks are special cases of the configuration models such that the degree distributions are specified by a Poisson distribution and $q_z=1$, respectively.

\subsection{Coordination games with a bilingual option}

Throughout the paper, we consider two types of activities, denoted by $A$ and $B$. Players will benefit from an activity if they enjoy the same activity as their neighbors. The two activities may be complements (e.g., drinking and smoking), substitutes (e.g., committing a crime and taking higher education), or neutral. For each of the two activities, players face a binary problem for which they adopt either action~0 (``don't  do it'') or action~1 (``do it''). 
The strategy set is thus given by $\{00,01,10,11\}\equiv S$. 
Strategy~11 is called the \emph{bilingual option} where the player engages in both activities~\citep{oyama2015bilingual,arigapudi2020bilingual}.
Each player selects a pure strategy $s\in S$, taking all the neighbors' actions as given. 

\subsubsection{Pure strategy equilibria in a bilateral game}

\begin{table}[tbh]
    \centering
    \caption{Payoff matrix of a coordination game with multiple activities. The two activities are complements (resp. substitutes) when $\delta>0$ (resp. $\delta<0$). }
    \begin{tabular}{cccccc}
     \multicolumn{1}{c}{}   &  &  \multicolumn{1}{c}{$00$}    & \multicolumn{1}{c}{$01$} & \multicolumn{1}{c}{$10$}&\multicolumn{1}{c}{$11$}  \\
    \cline{3-6}
    \multicolumn{1}{c}{}&  $00$ &\multicolumn{1}{|c}{$0,0$} & \multicolumn{1}{|c|}{$0,-c$} & \multicolumn{1}{|c|}{$0,-c$} &\multicolumn{1}{|c|}{$0,-2c$}\\
    \cline{3-6}
    &  \multicolumn{1}{c}{$01$} &\multicolumn{1}{|c}{$-c,0$} & \multicolumn{1}{|c|}{$a-c,a-c$} & \multicolumn{1}{|c|}{$-c,-c$} &\multicolumn{1}{|c|}{$a-c,a-2c$}\\
    \cline{3-6}
        \multicolumn{1}{c}{}&  $10$ &\multicolumn{1}{|c}{$-c,0$} & \multicolumn{1}{|c|}{$-c,-c$} & \multicolumn{1}{|c|}{$b-c,b-c$} &\multicolumn{1}{|c|}{$b-c,b-2c$}\\
    \cline{3-6}
     &  \multicolumn{1}{c}{$11$} &\multicolumn{1}{|c}{$\:\:\quad -2c,0\:\: \quad$} & \multicolumn{1}{|c|}{$a-2c,a-c$}&\multicolumn{1}{|c|}{$b-2c,b-c$}&\multicolumn{1}{|c|}{$a\!+\!b+\!\delta\! -\!2c,a\!+b\!+\!\delta\!-\!2c$}\\
    \cline{3-6}
    \end{tabular}
    \label{tab:two-good_payoff}
\end{table}
The payoff matrix of the $4\times 4$ bilateral coordination game is given in Table~\ref{tab:two-good_payoff}.
A player receives payoff $a>0$ (resp. $b>0$) of coordinating with a neighbor if they both join activity $A$ (resp. $B$). The cost of participating in an activity is given by $c$, where $0<c<a$ and $c<b$, so the net payoff of coordinating on activity $A$ (resp. $B$) leads to $a-c>0$ (resp. $b-c>0$). If two players successfully coordinate on both activities, both players receive $a+b$ plus extra payoffs $\delta$, for a total net payoff of $a+b+\delta-2c$.
$\delta$ represents the degree of complementarity between the two activities; $\delta>0$ (resp. $\delta<0$) when the two activities are complementary (resp. substitutes).\footnote{In \cite{oyama2015bilingual}, they consider an extra cost of taking a bilingual option, which needs to be incurred independently of the other player's response. In the current payoff structure, we do not consider such an extra cost, where the cost of taking a bilingual option is given by the sum of the costs for each action. Instead, we introduce an extra payoff, $\delta$, of coordinating on the bilingual option.} 
Note that if $a+b+\delta-2c> a-c$ and $a+b+\delta-2c > b-c$, which are satisfied when $\delta > c-b$ and $\delta > c-a$, then the pure-strategy Nash equilibria are given by $(00,00), (01,01),(10,10)$, and $(11,11)$, and $(11,11)$ is also the Pareto dominant equilibrium.
\begin{proposition}
Suppose that $a>c$, $b>c$, $\delta > c-a$, and $\delta > c-b$. The strategy pair $(11,11)$ is the Pareto dominant Nash equilibrium, and the (pairwise) risk dominant equilibrium is 
\begin{align}
\begin{cases}
 (00,00) & if\;\; c-\delta<a<2c,\; c-\delta<b<2c,\; and\; a+b+\delta <4c, \\
 (01,01) & if\;\;   2c<a,\; a>b,\;  and\;  c<b <2c-\delta, \\
 (10,10) & if\;\;   2c<b,\; b>a,\;  and\;  c<a <2c-\delta, \\
 (11,11) & if\;\;  2c-\delta<a,\;  2c-\delta<b,\;and\; a+b+\delta>4c.
\end{cases}\label{eq:risk_dominance}
\end{align}
\label{prop:rd_high_delta}
\end{proposition}
\begin{proof}
 It is obvious that the strategy pair $(11,11)$ is Pareto dominant.
 To obtain a risk dominance equilibrium, consider subgames restricted to two strategies. For $(01,01)$ to be risk dominant, for instance, $(01,01)$ has to risk dominate all the other Nash equilibria: $(00,00)$, $(10,10)$ and $(11,11)$~\citep{harsanyi1988general,young1993evolution}. Under the assumptions that $a>c$, $b>c$, $\delta>c-a$ and $\delta>c-b$, we have the following conditions of (strict) risk dominance for each subgame:
 \begin{align}
  & \text{$(01,01)$ risk dominates}   
   \begin{cases}
    (00,00) & \text{if}\; a > 2c,\\
    (10,10) & \text{if}\; a > b,\\
    (11,11) & \text{if}\; b < 2c-\delta,
  \end{cases} \label{eq:rd_01}\\
   &  \text{$(10,10)$ risk dominates}   
   \begin{cases}
    (00,00) & \text{if}\; b > 2c,\\
    (11,11) & \text{if}\; a< 2c-\delta,\\
  \end{cases}\label{eq:rd_10}\\
  & \text{$(11,11)$ risk dominates}\;\;   
    (00,00) \;\; \text{if}\; a+b+\delta > 4c.
 \label{eq:rd_11}
 \end{align}
 The conditions \eqref{eq:rd_01}--\eqref{eq:rd_11}, combined with the four assumptions, define pairwise risk-dominant equilibria depending on the relative sizes of the parameters, as presented in~\eqref{eq:risk_dominance}. 
\end{proof}
When $\delta$ is negative and large in absolute value such that $\delta<c-a$ or $\delta<c-b$, the action pair $(11,11)$ is no longer a Nash equilibrium since joining both activities is not beneficial. In the following analysis, we employ the assumption that $\delta>c-a$ and $\delta>c-b$ to focus on the situation in which the bilingual option can be an equilibrium strategy.
\begin{assumption}
$\delta>c-a$ and $\delta>c-b$.
\end{assumption}


\subsubsection{Multilateral games with $k$ neighbors}

In games with $k$ neighbors, the total payoffs of a player are given by the sum of the payoffs received in the $k$ bilateral games. 
Let $v(s,\bf{m})$ denote the total payoffs of a player who adopts strategy $s\in S$ and faces the neighbors' strategy profile ${\bf{m}}=(m_{00},m_{01},m_{10},m_{11})^\top$, where $m_{s}\in\mathbb{Z}_{\geq 0}$ denotes the number of neighbors that adopt strategy $s\in S$. Note that we have $\sum_{s\in S}m_s = k$ for nodes with degree $k$. For a given $\vecm$, the payoff of each strategy leads to  
\begin{align}
    v({00},{\bf{m}}) &= 0, \label{eq:v00}\\
    v({01},{\bf{m}}) &= -c(k-m_{01}-m_{11}) + (a-c)(m_{01}+m_{11}) \notag \\
                     &= -ck + a(m_{01}+m_{11}),   \label{eq:v01} \\
    v({10},{\bf{m}}) &= -ck + b(m_{10}+m_{11}),   \label{eq:v10}\\
    v({11},{\bf{m}}) &= -2ck + am_{01}+bm_{10} + (a+b+\delta)m_{11}.\label{eq:v11}   
\end{align}
The optimal strategy $s^*$ for a given $\vecm$, is then given by
\begin{align}
    s^*(\vecm) = \underset{s\in S}{\rm arg\,max}\; v(s,\vecm).
    \label{eq:optimal_s_coordination}
\end{align}
Since $v(00,\vecm)=0$, players adopt strategies other than $00$ if the total payoffs of those strategies are positive.\footnote{If there are tie values (i.e., $v(s,\vecm)=v(s',\vecm)$ for $s\neq s'$), we randomly select one strategy. If $k=0$, however, we keep the original strategy to focus on the peer effect.} The conditions for $v({01},\vecm)>0$ and $v({10},\vecm)>0$ are given by the following simple threshold rules:
\begin{align}
   v({01},{\bf{m}}) > 0 \hspace{8pt} {\rm{iff} }\hspace{8pt} \frac{m_{01}+m_{11}}{k} > \frac{c}{a}, \label{eq:v01_threshold} \\ 
   v({10},{\bf{m}}) > 0 \hspace{8pt} {\rm{iff} }\hspace{8pt} \frac{m_{10}+m_{11}}{k} > \frac{c}{b}.\label{eq:v10_threshold}
\end{align}
It should be noted that if there were only one activity, say activity $A$, the players' behavior would be ruled by a \emph{fractional threshold rule}, i.e., $m_{01}/k>c/a$, as in the well-studied models of contagion~\citep{morris2000contagion,Watts2002,Jackson2008book}, in which each player adopts either strategy $00$ or $01$. Since we have a bilingual option here, the optimal strategy is not determined simply by the two conditions \eqref{eq:v01_threshold} and \eqref{eq:v10_threshold}. 

The conditions $v(11,\vecm)>v(01,\vecm)$ and $v(11,\vecm)>v(01,\vecm)$ are respectively rewritten as
\begin{align}
      v({11},{\bf{m}}) > v(01,\vecm) \hspace{8pt} {\rm{iff} }\hspace{8pt} \frac{m_{10}}{k} + \left(1+\frac{\delta}{b}\right)\frac{m_{11}}{k}> \frac{c}{b}, \label{eq:v01_v11} \\
           v({11},{\bf{m}}) > v(10,\vecm) \hspace{8pt} {\rm{iff} }\hspace{8pt} \frac{m_{01}}{k} + \left(1+\frac{\delta}{a}\right)\frac{m_{11}}{k}> \frac{c}{a}.
     \label{eq:v10_v11}
\end{align}
Suppose for the moment that $\delta> 0$. Conditions~\eqref{eq:v10_threshold} and \eqref{eq:v01_v11} indicate that the fraction of neighbors joining activity~$B$ needed for strategy~$11$ to be preferable to strategy~$01$ will be less than that required for strategy~$10$ to be preferable to strategy~$00$. This is because when $\delta>0$, the payoff of engaging in both activities is greater than the sum of the payoffs of each activity.  
However, if the two activities are substitutes and thereby $\delta <0$, the benefit of joining an additional activity is diminished, so the condition for joining activity~$B$ \emph{in addition to} $A$ will be more stringent than condition~\eqref{eq:v10_threshold}, which is the condition for deciding whether to participate in $B$ or do nothing. 
The same argument also holds for the relationship between $v(11,\vecm)$ and $v(10,\vecm)$ (Eq.~\ref{eq:v10_v11}).

A crucial difference from the bilateral coordination games is that in this multilateral environment, coordinating with neighbors may not necessarily be the best response. For example, suppose that $m_{11}=0$, while $m_{01}$ and $m_{10}$ are sufficiently large such that Eqs.~\eqref{eq:v01_threshold} and \eqref{eq:v10_threshold} are satisfied (e.g., most friends use Windows or Mac OS, but none of them use both). Since $m_{11}=0$, conditions~\eqref{eq:v01_v11} and \eqref{eq:v10_v11} reduce to $m_{10}/k>c/b$ and $m_{01}/k>c/a$, respectively, and these conditions are also satisfied since Eqs.~\eqref{eq:v01_threshold} and \eqref{eq:v10_threshold} hold for $m_{11}=0$. Therefore, the best strategy turns out to be $s^*=11$ (e.g., using both Windows and Mac), regardless of the fact that no neighbor adopts $s=11$.

 It is also straightforward to show that 
 \begin{align}
      v(01,\vecm)>v(10,\vecm)\hspace{8pt} {\rm{iff} }\hspace{8pt}
      am_{01}-bm_{10}+(a-b)m_{11}>0,
      \label{eq:v01_10}
 \end{align}
 which suggests that a player's strategy may switch from $01$ to $10$ or vice versa in the process of diffusion, depending on the relative size of $m_{01}$ and $m_{10}$. The diffusion process is thus generally non-monotonic, unlike the standard (irreversible) binary-state cascade models~\citep{morris2000contagion,Watts2002,Jackson2008book,unicomb2021dynamics}.

\subsubsection{$p$-dominance and its relation to the threshold conditions}\label{sec:p-dominance}

In the previous subsection, we obtained the fractional threshold conditions for a strategy to be the best response. Here, we argue that a commonly used equilibrium selection criterion, \emph{$p$-dominance}~\citep{morris1995p-dominance,kajii1997robustness}, can be interpreted as a sufficient condition for the corresponding threshold condition.

\begin{definition}[\citealt{morris1995p-dominance}, \citealt{kajii1997robustness}]
Let $\tilde{u}(s_i,s_j)$ denote the $(i,j)$th element of the payoff matrix, and let $\pi$ on $S$ be a probability distribution where $\sum_{s\in S}\pi(s)=1$. The strategy pair $(s_i,s_i)$ is $p$-dominant if for every probability distribution $\pi$ on $S$ such that $\pi(s_i)\geq p$, and for every $s_{\ell}\in S$,
\begin{align}
    \sum_{s_j\in S}\pi(s_j)\tilde{u}(s_i,s_j) \geq \sum_{s_j\in S}\pi(s_j)\tilde{u}(s_\ell,s_j).
\end{align}
\label{def:p-dominance}
\end{definition}

A Nash equilibrium $(s_i,s_i)$ is said to be $p$-dominant if $s_i$ is a best strategy as long as the neighbors adopt strategy $s_i$ with a probability greater than or equal to $p$. While $p$-dominance is originally an equilibrium characterization for games with incomplete information, this concept is closely related to the threshold conditions shown in the previous section. 
From Eqs.~\eqref{eq:v01_threshold}, \eqref{eq:v01_v11} and \eqref{eq:v01_10}, one can show that $s=01$ will be the best response if the fraction of $01$-neighbors, $m_{01}/k$, satisfies all of the following three conditions: 
\begin{align}
    \frac{m_{01}}{k}>& -\frac{m_{11}}{k} + \frac{c}{a}, \label{eq:p-dom_01_00}\\
    \frac{m_{01}}{k}>& -\frac{a}{a+b}\frac{m_{11}}{k} + \frac{b}{a+b},\label{eq:p-dom_01_10}\\
    \frac{m_{01}}{k}> & 
    \begin{cases}
\frac{\delta}{b}\frac{m_{11}}{k}+1-\frac{c}{b},& \text{ if }\delta<0,\\
-\frac{\delta}{b+\delta}\frac{m_{10}}{k}+1-\frac{c}{b+\delta}, &\text{ if }\delta\geq0,
    \end{cases}\label{eq:p-dom_01_11}
\end{align}
where we used the inequality $m_{01}+m_{10}+m_{11}\leq k$. Now, consider more stringent conditions for which the inequalities \eqref{eq:p-dom_01_00}--\eqref{eq:p-dom_01_11} will hold for any configuration of $(m_{00}/k,m_{10}/k,m_{11}/k)$. Such conditions are given by
\begin{align}
 \frac{m_{01}}{k}> \frac{c}{a},\;\;\; 
    \frac{m_{01}}{k}>  \frac{b}{a+b},\;\;\;
    \frac{m_{01}}{k}>  
    \begin{cases}
1-\frac{c}{b}& \text{ if }\delta<0,\\
1-\frac{c}{b+\delta} &\text{ if }\delta\geq0.
    \end{cases}\label{eq:p-dom2}
\end{align}
The combined condition of \eqref{eq:p-dom2} is thus a sufficient condition for $s=01$ to be the best response. 
Indeed, the strategy pair $(s,s)$ is $p$-dominant for any $p\in[p_{s},1]$ and $s\in S$ such that
\begin{align}
    p_{00}&=\max\left\{\frac{a-c}{a},\frac{b-c}{b},1-\frac{2c}{a+b+\delta}\right\},\label{eq:p_00}\\
    p_{01}&=\max\left\{\frac{c}{a},\frac{b}{a+b},1-\frac{c}{\max\{b,b+\delta\}}\right\},\label{eq:p_01}\\
    p_{10}&=\max\left\{\frac{c}{b},\frac{a}{a+b},1-\frac{c}{\max\{a,a+\delta\}}\right\},\label{eq:p_10}\\
    p_{11}&=\max\left\{\frac{2c}{a+b+\delta},\frac{c}{b+\delta},\frac{c}{a+\delta}\right\}\label{eq:p_11}.
\end{align}
Note that $p_{01}$ is equivalent to the maximum threshold in~\eqref{eq:p-dom2}. This implies that if the fraction $m_{01}/k$ is interpreted as the probability of a randomly selected neighbor adopting $s=01$ (i.e., $\pi(01)$), then $p_{01}$ would coincide with the relevant threshold value in \eqref{eq:p-dom2}. This argument obviously holds true for the other strategies as well. Fig.~\ref{fig:p-dominance} shows the values of $p_{s}$ for different parameter combinations.

\begin{figure}
    \centering
    \includegraphics[width=15.8cm]{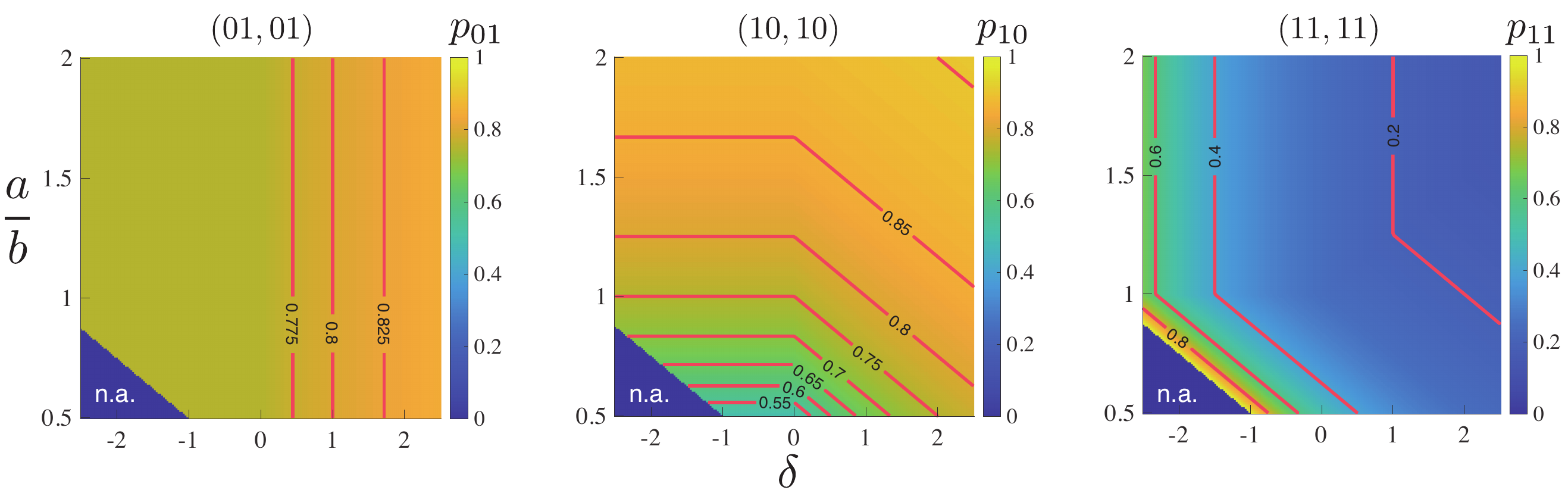}
    \caption{Minimum values of $p$ for a strategy pair to be $p$-dominant. We set $b=4$ and $c=1$, focusing on the parameter space in which $\delta>c-a$ and $\delta>c-b$.}
    \label{fig:p-dominance}
\end{figure}

In section~\ref{sec:quantitative}, we calculate the dynamical paths for the aggregate share of each strategy. There, we will examine whether a strategy $s$ would increase its popularity once its share exceeds $p_s$ indicated in Fig.~\ref{fig:p-dominance}. If the probability distribution of neighbors' strategies is well approximated by the aggregate shares of each strategy in the population, then $p$-dominance would indeed become a sufficient condition for a strategy to be adopted by the majority of players in the steady state.

\section{Analytical framework for describing the dynamics}

Now we present analytical approaches to describing the dynamics of diffusion. 
Throughout the analysis, we assume that strategy~$00$ is the status quo for all players, except for a certain fraction of players who are initially ``active.'' Then, the neighbors of the initially active players (or the ``seed players'') may respond by adopting optimal strategies other than $00$, which may initiate cascades through edges in the network.

Since the strategic choice of each player can be expressed as a function of the neighbors' profile $\vecm$, we introduce a probability $F_{\vecm}(s\to s')$ in which the strategy changes from $s$ to $s'$:\footnote{Players' choice also depends on degree $k$ in the coordination games, but the degree is directly obtained by $\vecm$ since $k = \sum_{s\in S}m_s$.}
\begin{align}
    F_{\vecm}(s\to s') =
    \begin{cases}
    1  \;\;{\text{ if }}\; s' = s^*(\vecm),  \\ 
    0 \;\;{\text{ otherwise}}.
    \end{cases}
    \label{eq:response_func}
\end{align}
 Note that the choice of new strategy does not depend on the player's current strategy $s$. This indicates that, unlike the standard binary-state models~\citep{morris2000contagion,Watts2002}, the strategic choice is fully reversible and is determined solely in response to the neighbors' profile $\vecm$. We also call $F_{\vecm}$ the \emph{response function}. 

\subsection{Approximate master equations}

Let $\mathcal{V}_{|\vecm|=k}^s(t)$ denote the set of $k$-degree players that belong to the $(s,\vecm)$ class at time $t$, i.e., all the players in $\mathcal{V}_{|\vecm|=k}^s$ adopt strategy $s\in S$ and have $k$ neighbors. The fraction of $k$-degree players belonging to the $(s,\vecm)$ class is given by 
\begin{align}
    \rho_{k,\vecm}^s(t) \equiv \frac{|\mathcal{V}_{|\vecm|=k}^s(t)|}{q_k N},
\end{align}
where $q_k N$ is the total number of players with degree $k$. Then, the expected fraction of players adopting strategy $s\in S$ (i.e., $s$-players) at time $t$ is given as
\begin{align}
    \rho^s(t) = \sum_k q_k\sum_{|\vecm|=k} \rho_{k,\vecm}^s(t),
    \label{eq:rho_s}
\end{align}
where $\sum_{|\vecm|=k}$ denotes the sum over $\vecm =$ $(m_{00},m_{01},m_{10},m_{11})$ such that $\sum_{s\in S} m_s = k$. Our interest in this section is to calculate the dynamical path of $\rho^s$ for each strategy $s\in S$. To this end, we need to describe the dynamics of $\rho_{k,\vecm}^s$, which capture changes in the population in a given $(s,\vecm)$ class. 

There are four factors that change $\rho_{k,\vecm}^s$ over time.  Players will \emph{leave} the $(s,\vecm)$ class if \emph{i}) their strategy changes from $s$ to $s'(\neq s)$, or \emph{ii}) their neighbor profile changes from $\vecm$ to $\vecm' (\neq \vecm)$. On the other hand, players will newly \emph{enter} the $(s,\vecm)$ class if \emph{iii}) the players' strategies shift from $s'(\neq s)$ to $s$, or \emph{iv}) the neighbor profiles shift from $\vecm'(\neq \vecm)$ to $\vecm$.
To take into account all of these four factors that will affect the behavior of $\rho_{k,\vecm}^s$, we employ an approximate master equations (AMEs) approach~\citep{gleeson2011high,gleeson2013binary,fennell2019multistate}. The dynamics of $\rho_{k,\vecm}^s$ is given by the following differential equation:
\begin{align}
     \frac{d}{dt} \rho_{k,\vecm}^s \;=\; & -\overbrace{\sum_{s' \neq s}F_{\vecm}(s\to s')\rho_{k,\vecm}^s}^{\text{i) Leave, $s\!\to$}} 
     \;\: -\overbrace{\sum_{r\in S}\sum_{r'\neq r}m_{r}\phi_{s}(r\to r')\rho_{k,\vecm}^s}^{\text{ii) Leave, $\vecm\!\to$}} \notag \\
     & +\;\: \underbrace{\sum_{s' \neq s}F_{\vecm}(s'\to s)\rho_{k,\vecm}^{s'}}_{\text{iii) Enter, $\to\! s$}} 
    \;\: +\;\: \underbrace{\sum_{r\in S}\sum_{r'\neq r}(m_{r'}+1)\phi_{s}(r'\to r)\rho_{k,\vecm-{\bf e}_{r} + {\bf e}_{r'}}^s}_{\text{iv) Enter, $\to\!\vecm$}},
\label{eq:AME_eq}
\end{align}
for all $s\in S$ and $\vecm$ such that $\sum_{s\in S}m_s=k\in \{1,\ldots,k_{\rm max}\}$. We assume that $\frac{d}{dt}\rho_{k,\vecm}^s=0$ for $k=0$, meaning that isolated players will not be influenced by other players. $\phi_s(r\to r')$ denotes the probability that a neighbor of an $s$-player changes the strategy from $r$ to $r'$ for $r,r'\in S$:
\begin{align}
    \phi_s(r\to r') = \frac{\sum_k q_k \sum_{|\vecm|=k} m_s\rho_{k,\vecm}^rF_\vecm(r\to r')}{\sum_k q_k \sum_{|\vecm|=k} m_s\rho_{k,\vecm}^r},
    \label{eq:phi_s}
\end{align}
where the denominator $\sum_k q_k \sum_{|\vecm|=k} m_s\rho_{k,\vecm}^r$ represents the expected number of $(s)$--$(r)$ edges that connect $s$-players and $r$-players at a given point in time. The expected number of $(s)$--$(r)$ edges that change to $(s)$--$(r')$ in a small time interval $dt$ is then given as $\sum_k q_k \sum_{|\vecm|=k} m_s\rho_{k,\vecm}^rF_\vecm(r\to r')dt$. The probability of a  $(s)$--$(r)$ edge being changed to a $(s)$--$(r')$ edge in the interval $dt$, denoted by $\phi_s(r\to r')dt$, is calculated as the ratio of the expected number of edges that changes from $(s)$--$(r)$ to $(s)$--$(r')$ and the expected number of $(s)$--$(r)$ edges, which leads to Eq.~\eqref{eq:phi_s}.
 
The first and second terms in Eq.~\eqref{eq:AME_eq} respectively capture the aforementioned factors i) and ii). The first term captures the rate at which the strategy of a player in the $(s,\vecm)$ class shifts from $s$ to $s'(\neq s)$ in an infinitesimal time interval $dt$. Similarly, the second term denotes the rate at which the neighbors' profile differs from $\vecm$.
The third and fourth terms correspond to the factors iii) and iv), respectively. The third term captures the rate at which the strategy of the $(s',\vecm)$-class players changes from $s'(\neq s)$ to $s$.
The fourth terms shows the rate that the neighbors' profile newly becomes $\vecm$.
${\bf e}_r$ denotes the $4\times 1$ vector with $1$ in the $r$-th element and $0$ in the other elements.
The expression $\vecm-{\bf e}_r+{\bf e}_{r'}$ thus represents the neighbor profile that has $m_{r'}+1$ in the $r'$-th element and $m_{r}-1$ in the $r$-th element.
It should be noted that Eq.~\eqref{eq:AME_eq} describes the dynamics under asynchronous updates in which only a fraction $dt$ of players can change their strategies in response to their neighbor profiles in a small time interval. 
Therefore, while the optimal strategy is given as a function of $\vecm$ (Eq.~\ref{eq:optimal_s_coordination}), the current strategy of a player does not necessarily have a one-to-one correspondence with $\vecm$ in the process of diffusion. $s$ and $\vecm$ recover the one-to-one correspondence defined by Eq.~\eqref{eq:optimal_s_coordination} in the steady state at which players have no incentive to change their strategies, i.e., a Nash equilibrium.

The system of differential equations can be solved by providing initial values $\rho_{k,\vecm}^s(0)$ for all $s\in S$ and $\vecm\in\{\vecm: |\vecm|=k,\:k= 0,\ldots,k_{\rm max}\}$. 
The number of differential equations in the system is calculated as the total number of ways of picking $k$ balls in an urn with replacement and without ordering. In the urn there are balls of four different colors, so the total number of color patterns when one picks $k$ balls is given by $\binom{4+k-1}{k} = \binom{k+3}{k}$. Since a player selects one of the four strategies and the degree $k$ ranges from $0$ to $k_{\rm max}$, the total number of differential equations leads to $4\sum_{k=0}^{k_{\rm max}}\binom{k+3}{k}$.\footnote{We solve the system of differential equations using an ODE solver, \texttt{ode45}, for Matlab. Our Matlab code is based on the multi-state-SOLVER package available from \texttt{https://github.com/peterfennell/multi-state-SOLVER}.}  

\subsection{Mean-field approximation}

In the AME approach, it is assumed that the transition rate of a neighbor's strategies, $\phi_s(r\to r')$, is independent of the states of the other neighbors, and their profile $\vecm$ is used to characterize each player's state. In the mean-field (MF) approach, we impose a stronger assumption that the strategies of neighbors are independent and randomly distributed following a multinomial distribution. Thus, the neighbor profile $\vecm$ is ignored and each player's state is characterized by a combination of strategy and degree, $(s,k)$, rather than $(s,\vecm)$.  

 To calculate the average fraction of players belonging to the $(s,k)$ class, we define $\rho_{k}^s$ as the sum of $\rho_{k,\vecm}^s$ over $\vecm$:
 \begin{align}
     \rho_{k}^s(t) \equiv \sum_{|\vecm|=k}\rho_{k,\vecm}^s(t).
     \label{eq:rho_sk}
 \end{align}
The average probability that a neighbor of a player adopts strategy $r\in S$, denoted by $\omega^r$, is given by
\begin{align}
    \omega^r(t) \equiv \sum_{k\geq 1}\frac{k q_k}{z}\rho_{k}^r(t),
    \label{eq:omega_MF}
\end{align}
where $kq_k/z$ is the probability that a randomly selected neighbor has degree $k$. The dynamics of $\rho_k^s$ are then expressed as
\begin{align}
    \frac{d}{dt}\rho_{k}^s = -\sum_{s'\neq s}\rho_k^s\sum_{|\vecm|=k}\mathcal{M}_{k}(\vecm,\vecw)F_\vecm(s\to s') 
     +\sum_{s'\neq s}\rho_k^{s'}\sum_{|\vecm|=k}\mathcal{M}_{k}(\vecm,\vecw)F_\vecm(s'\to s),
     \label{eq:MF_eq}
\end{align}
where $\vecw \equiv (\omega^{00},\omega^{01},\omega^{10},\omega^{11})^\top$, and $\mathcal{M}_k(\vecm,\vecw)$ is the multinomial distribution:
\begin{align}
    \mathcal{M}_k(\vecm,\vecw) \equiv \frac{k!}{m_{00}!m_{01}!m_{10}!m_{11}!}(\omega^{00})^{m_{00}}(\omega^{01})^{m_{01}}(\omega^{10})^{m_{10}}(\omega^{11})^{m_{11}}.
\end{align}
The first term in Eq.~\eqref{eq:MF_eq} captures the rate at which a player leaves the $(s,k)$ class by changing the strategy from $s$ to $s'(\neq s)$. The second term denotes the rate at which $(s',k)$-class players newly employ strategy $s$. Note that since we have four strategies and the degree ranges from $0$ to $k_{\rm max}$, there are $4(k_{\rm max}+1)$ differential equations in the MF scheme, where $\frac{d}{dt}\rho_0^s = 0$. 

\subsection{Dominant strategies in Nash equilibria}

\begin{figure}[tb]
    \centering
    \includegraphics[width=15.7cm]{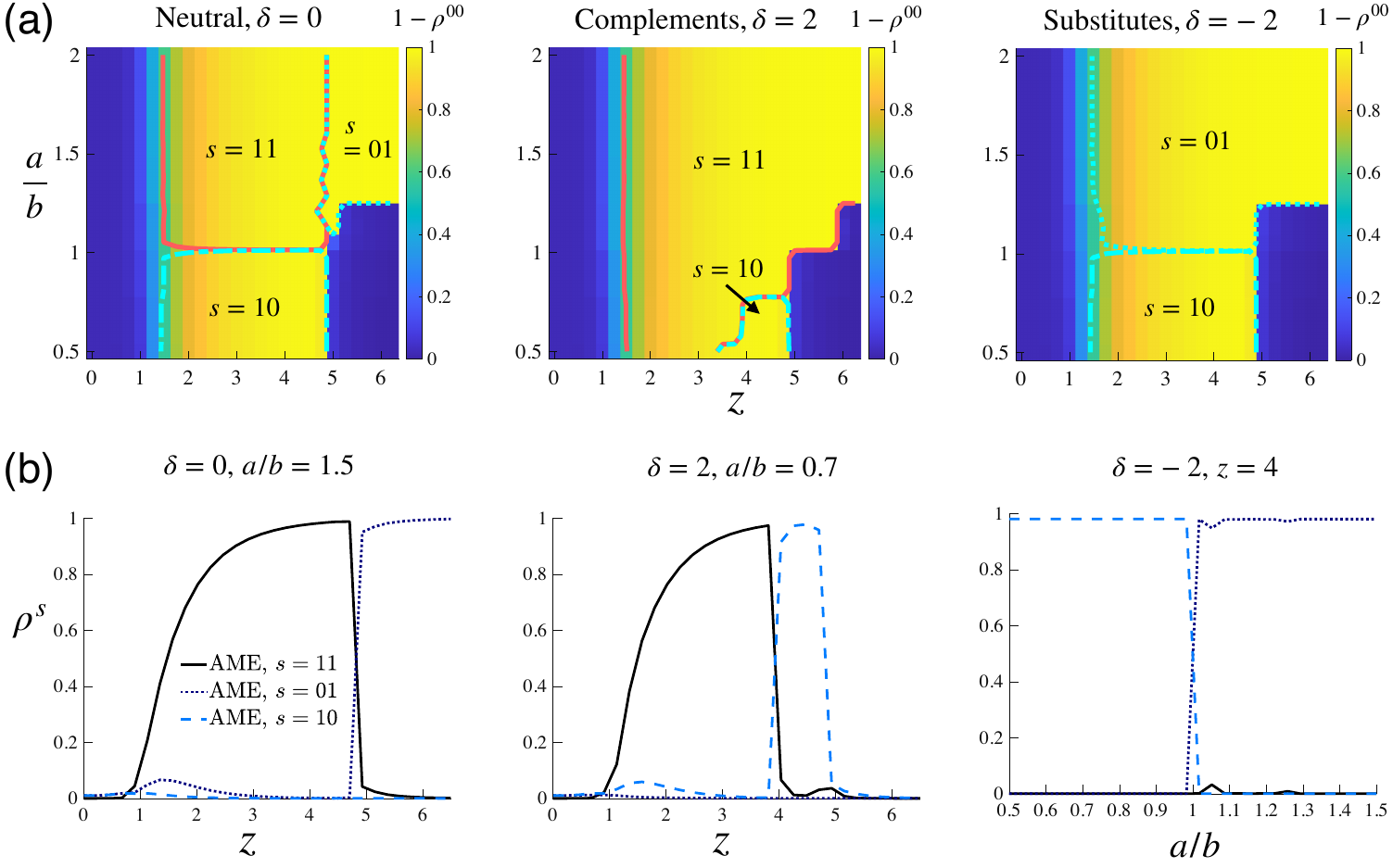}
    \caption{Dominant strategy in Nash equilibrium. (a) Dominant strategy, which is a strategy employed by at least 50\% of players, can be either $01$ (light-blue dotted), $10$ (light-blue dashed), or $11$ (red solid). Color denotes $1-\rho^{00}(T)$ calculated by the AME.  We use \ER networks with mean degree $z$. (b) Fraction of players adopting strategy~$s$ in the steady state. In both (a) and (b), we set $b=4$, $c=1$, $\rho_0^{01}=\rho_0^{10}=0.01$, $\rho_0^{11}=0$, and $T=50$. }
    \label{fig:heatmap}
\end{figure}

Let us examine how the steady state of diffusion processes is affected by the payoff parameters $(a,b,c)$, substitutability parameter $\delta$, and network connectivity $z$. Fig.~\ref{fig:heatmap}a shows the cascade region calculated using the AME. 
We find that within the cascade region in which $1-\rho^{00}> 0$, there are some distinct subregions characterized by different dominant strategies, where we define a dominant strategy as the strategy employed by more than 50\% of players in the Nash equilibrium (Fig.~\ref{fig:heatmap}a).
For a given $\delta$, the dominant strategy varies with the relative attractiveness, $a/b$, and the average degree, $z$. When $\delta=0$, for instance, strategy~$01$ (resp. strategy~$10$) prevails only when the payoff value $a$ (resp. $b$) is relatively large. 
Which strategy dominates the others also depends  strongly on the degree of substitutability $\delta$. Intuitively, strategy~$11$ will be dominant in most of the cascade region when both activities are complements (Fig.~\ref{fig:heatmap}a, \emph{middle}), whereas either strategy~$01$ or $10$ will be adopted when the two activities exhibit strong substitutability (Fig.~\ref{fig:heatmap}a, \emph{right}). 

It should be noted that shifts in dominant strategies, or \emph{phase transitions}, occur in a discontinuous manner at the boundary of the dominant regions. A small change in a parameter may cause the current dominant strategy to be discarded in favor of another. (Fig.~\ref{fig:heatmap}b). In the class of binary-state threshold models, the possibility of global cascades can arise only when the network connectivity (i.e., the average degree $z$) is neither too weak nor too strong, and the boundaries of the average degree at which phase transitions occur are called \emph{critical points}~\citep{Watts2002,GaiKapadia2010}.
In our ``multistate'' cascade model, there are possibly six types of boundaries at which dominant strategies switch: $(00,01)$, $(00,10)$, $(00,11)$, $(01,10)$, $(01,11)$, and $(10,11)$. The critical points for these transition patterns are characterized not only by the network connectivity $z$, but also by the relative attractiveness $a/b$ and the degree of substitutability $\delta$.

\subsection{(In)stability of diffusion dynamics and symmetry breaking}\label{sec:symmetry_breaking_coordination}

To describe the mechanics of diffusion dynamics using phase diagrams, here we employ the MF approach~\eqref{eq:MF_eq}, assuming that networks are $z$-regular random graphs. Note that when the network is $z$-regular, we have $\rho_k^s(t) = \rho^s(t)=\omega^s(t)$ from Eqs.~\eqref{eq:rho_s}, \eqref{eq:rho_sk} and \eqref{eq:omega_MF}, since $q_k = 1$ for $k=z$ and $q_k=0$ otherwise.
Thus, the system of MF equations consists of four equations for the four variables: $\rho^{00}(t)$, $\rho^{01}(t)$, $\rho^{10}(t)$ and $\rho^{11}(t)$. Note that since there is a constraint that $\sum_{s\in S} \rho^s(t)=1$, we can focus on three of them by replacing $\rho^{00}$ ($=\omega^{00}$) with $1-\rho^{01}-\rho^{10}-\rho^{11}$.
In the following analysis, we mainly describe the dynamics of $\rho^{01}$, $\rho^{10}$ and $\rho^{11}$, where $\rho^{00}$ is determined residually.\footnote{We drop the time subscript $t$ for brevity.} 

\begin{figure}[tb]
    \centering
    \includegraphics[width=15.5cm]{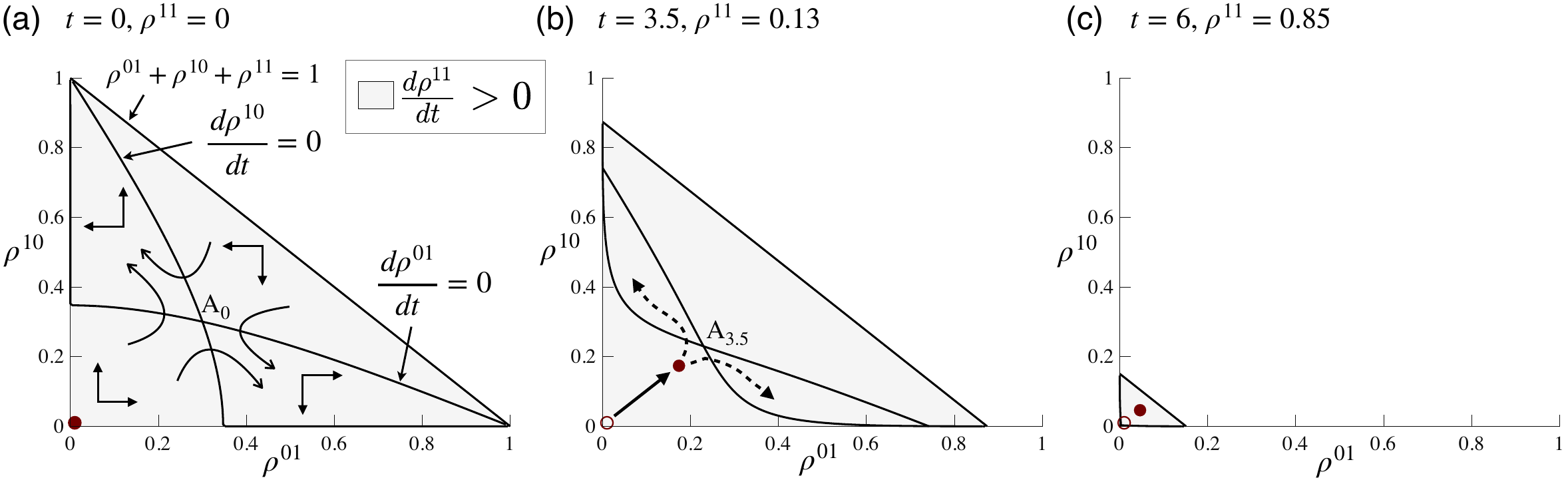}
    \caption{Phase diagram for the diffusion of neutral activities ($\delta=0$) obtained by the MF method using random 4-regular networks. Note that we have $d\rho^{11}/dt>0$ independently of $\rho^{01}$ and $\rho^{10}$, so that $\lim_{t\to\infty}\rho^{01}(t)=\lim_{t\to\infty}\rho^{10}(t)=0$. Temporal symmetry breaking may occur since ${\rm A}_t$ is a saddle point. $a=b=4$, $c=1$, $\rho_0^{01}=\rho_0^{10}=0.01$ and $\rho_0^{11}=0$.}
    \label{fig:symmetry_breaking_delta0}
\end{figure}

Fig.~\ref{fig:symmetry_breaking_delta0} shows the (sliced) phase diagrams for the diffusion of neutral goods (i.e., $\delta=0$) in the $\rho^{01}$--$\rho^{10}$ plane. Note that the combination of $(\rho^{01},\rho^{10})$ must lie within the feasible region $\{(\rho^{01},\rho^{10}):  \rho^{01}+\rho^{10}+\rho^{11}\leq 1,\: \rho^{s}\in [0,1] \}$ (shaded in gray). The feasible region is divided into four subregions by the signs of the time derivatives $d\rho^{01}/dt$ and $d\rho^{10}/dt$ (denoted by arrows), which are calculated using Eq.~\eqref{eq:MF_eq} for a given $(\rho^{01},\rho^{10},\rho^{11})$. Suppose that the initial state is given by $(\rho^{01}(0),\rho^{10}(0),\rho^{11}(0))=(0.01,0.01,0)$, which is denoted by red circle in Fig.~\ref{fig:symmetry_breaking_delta0}a. Since the time derivatives of $\rho^{01}$ and $\rho^{10}$ are both positive at
the initial point, the combination $(\rho^{01},\rho^{10})$ moves toward the saddle point ${\rm A}_0$. Along with this, since $d\rho^{11}/dt>0$, the feasible region shrinks as $\rho^{11}$ increases with time.  The closer the combination $(\rho^{01},\rho^{10})$ is to the saddle point, the more likely it is that the symmetry between $\rho^{01}$ and $\rho^{10}$ will be broken (Fig.~\ref{fig:symmetry_breaking_delta0}b). 
Note that while the theoretical path of $(\rho^{01},\rho^{10})$ suggested by the MF equations approaches the saddle point, the actual (or simulated) path could deviate from the theoretical path because seed players are not necessarily located in symmetric positions. The extent to which a seed player affects the other players' behavior would depend not only on the number of their direct neighbors, but also on the number of neighbors at two or more steps away. Thus, there is no guarantee that the realized fraction of $s$-players is equal to the theoretical average $\rho^{s}$, as there are some fluctuations in the network structure that could not be captured by the current ``average-based'' approximation methods.
If the realized path in a given network deviates from the MF path at least to some extent, then the subsequent path will move further away from the saddle point. 
However, because $\rho^{11}$ is increasing, the feasible region continues to shrink, resulting in $\lim_{t\to\infty}\rho^{01}(t) = \lim_{t\to\infty}\rho^{01}(t)= 0$ and $\lim_{t\to\infty}\rho^{11}(t)=1$ (Fig.~\ref{fig:symmetry_breaking_delta0}c). This indicates that the observed symmetry breaking, if any, is transient, and strategy~$11$ will always be dominant in the steady state.

\begin{figure}[tb]
    \centering
    \includegraphics[width=15.5cm]{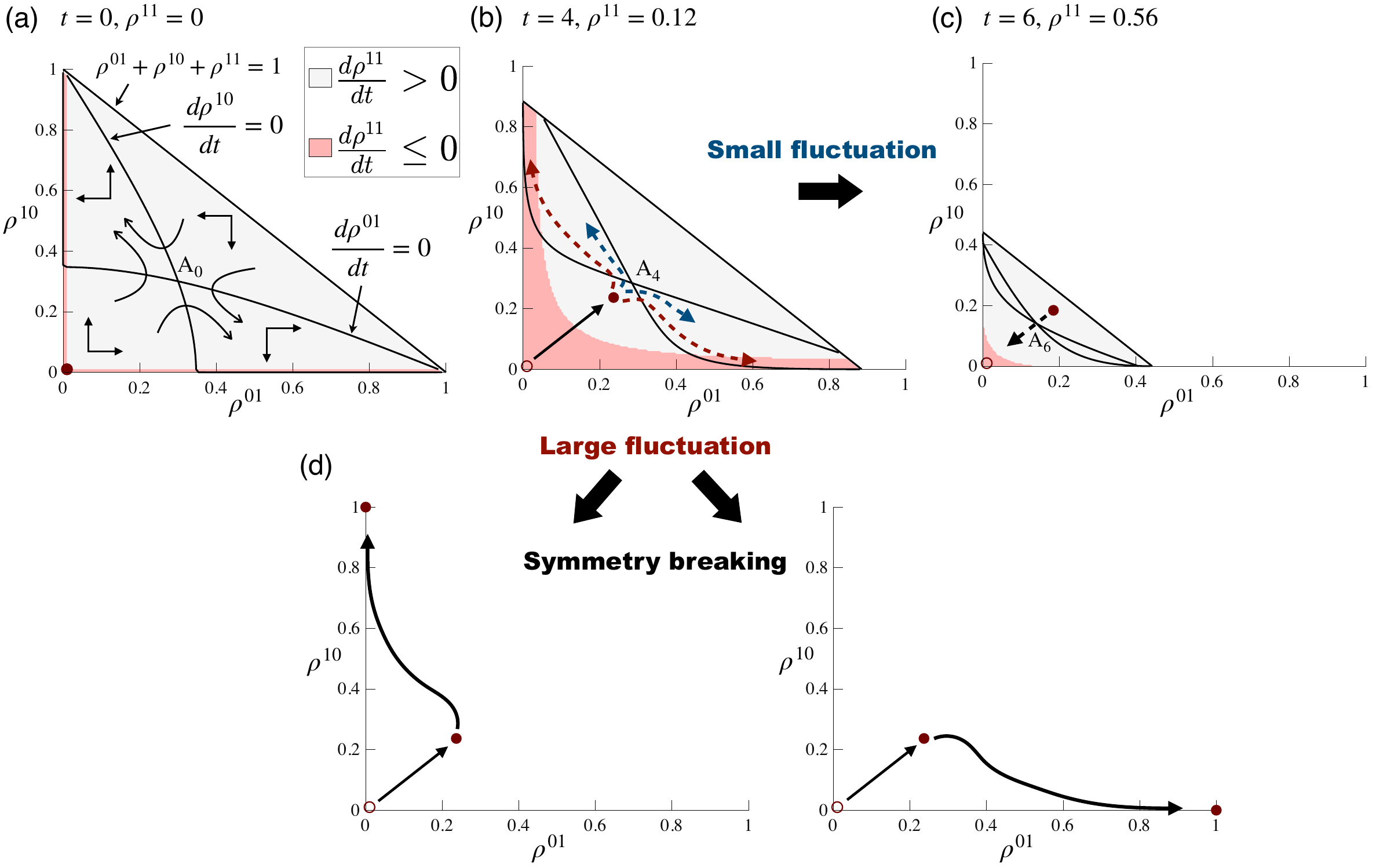}
    \caption{Phase diagram for the diffusion of substitutes ($\delta=-2$) obtained by the MF method using random 4-regular networks.  If the combination of $(\rho^{01},\rho^{10})$ reaches the pale-red region in which $d\rho^{11}/dt<0$, the feasible region of $(\rho^{01},\rho^{10})$ expands, and the share of strategies approaches a stable state: $(\rho^{00},\rho^{01},\rho^{10},\rho^{11}) = (0,1,0,0)$ or $(0,0,1,0)$. $a=b=4$, $c=1$, $\rho_0^{01}=\rho_0^{10}=0.01$ and $\rho_0^{11}=0$.}
    \label{fig:symmetry_breaking}
\end{figure}

Fig.~\ref{fig:symmetry_breaking} illustrates the phase diagrams for substitutes ($\delta=-2$). We see that there is still a saddle point as seen in Fig.~\ref{fig:symmetry_breaking_delta0}, but there arises a region in which $d\rho^{11}/dt<0$ (shaded in pale red). In the diffusion of neutral activities, the feasible region always shrinks and thus any deviation from the theoretical path will be diminished, leading to the unique equilibrium $(\rho^{01},\rho^{10})=(0,0)$. 
In the diffusion of substitutes, the same mechanics will still hold if the deviation from the MF path is not sufficiently large (Fig.~\ref{fig:symmetry_breaking}c). However, the feasible region will expand if the deviation from the theoretical path is sufficiently large such that $d\rho^{11}/dt<0$ (Fig.~\ref{fig:symmetry_breaking}b). If this is the case, the observed symmetry breaking will no longer be a transient phenomenon, where the equilibrium for $(\rho^{01},\rho^{10})$ will be given by $(1,0)$ or $(0,1)$ (Fig.~\ref{fig:symmetry_breaking}d). In the following section, we will show that such a persistent symmetry breaking can indeed occur not only in $z$-regular random networks, but also in more complex \ER networks.

\section{Quantitative analysis}\label{sec:quantitative}

To see how well the differential equations obtained by the AME (Eq.~\ref{eq:AME_eq}) and MF (Eq.~\ref{eq:MF_eq}) approaches capture the dynamics of diffusion, we compare the analytical results with simulated diffusion processes. We use \ER networks in the baseline analyses. The procedure of numerical simulation is as follows:
\begin{enumerate}
     \item For a given $z$ and $N$, generate an \ER network with connecting probability $q=z/(N-1)$. 
     \item Select seed players at random so that there are $\lfloor \rho_0^{01} N\rfloor$ players adopting strategy~$01$ and $\lfloor \rho_0^{10} N\rfloor$ players adopting strategy~$10$. The other players employ strategy~$00$ as the status quo.
     \item Choose a fraction $dt\in (0,1)$ of players uniformly at random and update their strategies to maximize their payoff $v$.
     \label{it:async}
     \item Repeat step 3 until convergence, where further updates would not change the strategy of any player.  
     \item Repeat steps 1--4.
 \end{enumerate}
To conduct simulations in a manner consistent with the continuous-time framework, we implement an asynchronous update in step~\ref{it:async}, where a randomly chosen fraction $dt$ of the players updates their strategies in an infinitesimally small interval $dt$~\citep{gleeson2011high,gleeson2013binary}. This is consistent with the AME and MF formulations because changes in the fraction $d\rho^s_{k,\bf{m}}$ in Eq.~\eqref{eq:AME_eq} (resp. $d\rho^s_{k}$ in Eq.~\ref{eq:MF_eq}) are explained by a fraction $dt$ of the possible shifts in  the players' states captured by the RHS of Eq.~\eqref{eq:AME_eq} (resp. Eq.~\ref{eq:MF_eq}). We set $dt=0.01$ in all simulations.

\subsection{Dynamics of diffusion: symmetric payoffs}

\begin{figure}[tb]
    \centering
    \includegraphics[width=15.8cm]{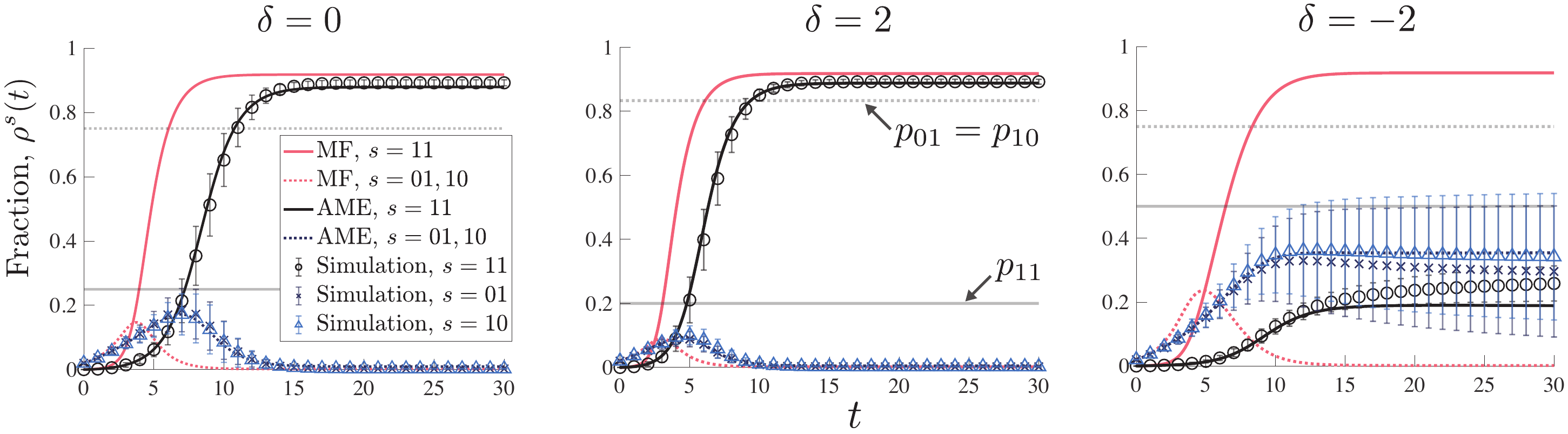}
    \caption{Theoretical and simulation results for the path of $\rho^{s}(t)$. Symbols denote the averages over 100 runs on \ER networks with $z=2.5$ and $N=3000$. Error bar denotes one standard deviation. Horizontal dotted and solid gray lines respectively denote the values of $p$ for $(01,01)$ (or equivalently $(10,10)$) and $(11,11)$ to be $p$-dominant. 
    $a=b=4$, $c=1$, $\rho_0^{01}=\rho_0^{10}=0.02$ and $\rho_0^{11}=0$.}
    \label{fig:path_vs_deltas}
\end{figure}

Let us first consider the case of symmetric activities, where $a=b$ and $\rho^{01}(0) = \rho^{10}(0)$. Since $a=b$, the two activities $A$ and $B$ are equally attractive, so it is expected that $\rho^{01}(t) = \rho^{10}(t)$ for all $t$, other things being equal.  
We find that when the two activities are neutral ($\delta=0$) or complements ($\delta>0$),  the path of $\rho^s(t)$ obtained by the AME method well matches the simulated path for any strategy~$s$, while the MF approximation is generally less accurate (Fig.~\ref{fig:path_vs_deltas}, \emph{left} and \emph{middle}). The standard deviations of the simulated $\rho^s(t)$ over 100 runs become vanishingly small as $t\to\infty$.

In contrast, when the two activities are substitutes ($\delta<0$), the standard deviations of the simulated $\rho^s(t)$ increase over time, while the average values are still well approximated by the corresponding theoretical values obtained by the AME (Fig.~\ref{fig:path_vs_deltas}, \emph{right}).
The discrepancy between theory and simulation observed in the propagation of substitutes reflects the fact that one of the two activities occasionally dominates the other even if the payoff parameters and the seed fractions are perfectly symmetric. This symmetry breaking occurs through the mechanism described in section~\ref{sec:symmetry_breaking_coordination}; when the activities are substitutable, players would not have an incentive to choose strategy~$11$, so once cascade occurs, either strategy~$01$ \emph{or} $10$ would prevail. The AME and MF solutions suggest the existence of a symmetric equilibrium, but such an equilibrium would not be achieved in numerical simulations where players are connected in a heterogeneous way. 

Fig.~\ref{fig:path_vs_deltas} also shows the values of $p$ such that strategy pair $(s_i,s_i)$ is $p$-dominant (denoted by $p_{s_i}$ in Fig.~\ref{fig:path_vs_deltas}, \emph{middle}). As argued in section~\ref{sec:p-dominance}, if the probability distribution of neighbors' strategies is approximated by the aggregate shares of each strategy, i.e., $\pi(s)\approx \rho^{s}$ for all $s\in S$, then $p_{s_i}$ would be roughly interpreted as a (sufficient) threshold of $\rho^{s_i}$ above which the strategy $s_i$ will prevail. We see from Fig.~\ref{fig:path_vs_deltas} that in the cases of $\delta=0$ and $2$, in which $s=11$ prevails, $\rho^{01}$ does not exceed $p_{01}$ while strategy~$01$ gained some extent of popularity when $t$ is between $5$ and $8$. On the other hand, $\rho^{11}$ exceeds $p_{11}$ around those times and continued to gain popularity until it reaches the steady state. When $\delta=-2$, in contrast, no strategy exceeds the corresponding $p$ to be adopted by the majority of players.

\begin{figure}[tb]
    \centering
    \includegraphics[width=15.8cm]{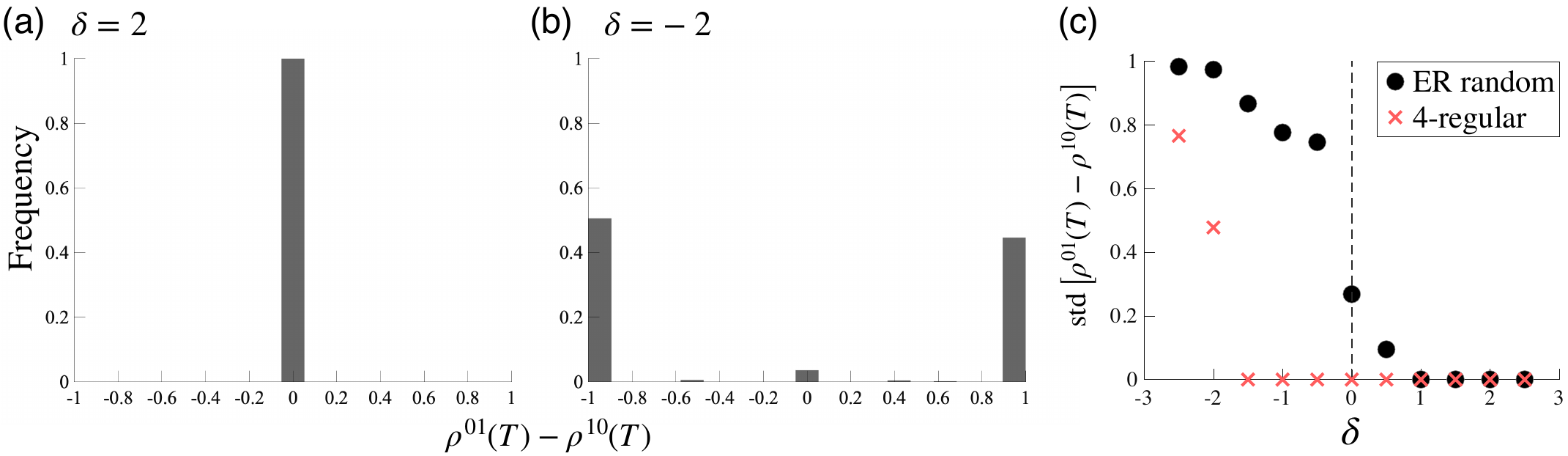}
    \caption{Histogram of the difference $\rho^{01}(T)-\rho^{10}(T)$ for \ER networks with (a) $\delta=2$ and (b) $\delta=-2$. Note that we always have $\rho^{01}(T)=\rho^{10}(T)$ when the activities are complements (i.e., $\delta=2$), but when they are substitutes (i.e., $\delta=-2$), we have either $(\rho^{01}(T),\rho^{10}(T))=(1,0)$ or $(0,1)$ in most of the 500 simulation runs. (c) Standard deviation of $\rho^{01}(T)-\rho^{10}(T)$ over 100 simulations. Black circle denotes the case of \ER random networks, and red cross denotes the case of random $z$-regular networks. In both panels, we set $z=4$, $\rho_0^{01}=\rho_0^{10}=0.01$, $\rho_0^{11}=0$, $a=b=4$, $c=1$ and $T=100$.}
    \label{fig:symmetry_breaking_dist}
\end{figure}

Figs.~\ref{fig:symmetry_breaking_dist}a and \ref{fig:symmetry_breaking_dist}b show the distributions of the difference $\rho^{01}(T)-\rho^{10}(T)$ for different values of $\delta$. When the two activities are complements, we always have $\rho^{01}(T)=\rho^{10}(T)$ (Fig.~\ref{fig:symmetry_breaking_dist}a), in which case the AME solution is accurate for every instance of simulated contagion processes.
In contrast, when they are substitutes, we have either $(\rho^{01}(T),\rho^{10}(T))=(1,0)$ or $(0,1)$ in most of the simulation runs. This suggests that either activity~$A$ or $B$ dominates with probability $\approx 0.5$, depending on the details of the structure of the network that  the AME or MF methods do not capture. Note that the AME solution will still describe the simulated equilibrium \emph{on average}, but it does not necessarily mean that the AME solution is accurate for every instance of the propagation processes. In general, there is a negative relationship between $\delta$ and the standard deviation of $\rho^{01}(T)-\rho^{10}(T)$, and the essential result also holds true for random regular networks (Fig.~\ref{fig:symmetry_breaking_dist}c).

\subsection{Dynamics of diffusion: asymmetric payoffs}\label{sec:asymmetric_coordination}

 As we saw in the previous section, symmetry-broken diffusion is occasionally observed when the degree of substitutability $\delta$ is low (Fig.~\ref{fig:symmetry_breaking_dist}), in which case the approximation methods do not accurately predict the path of simulated $\rho^s(t)$. However, when there is a certain extent of intrinsic asymmetry between the two activities, the AME method works quite well even for substitutable activities. If $a>b$, for example, activity~$A$ is always preferred to activity~$B$, other things being equal, and accordingly the number of $10$-players will be smaller than the number of $01$-players. This makes it difficult to have a situation where the total payoff $v(10)$ is comparable to $v(01)$, which would be achieved if a certain fraction of neighbors adopt strategy~$10$ while a smaller fraction of neighbors adopts strategy~$01$. Thus, the unstable nature of the diffusion process that we see for equally attractive activities will not materialize when the intrinsic attractiveness is different.
 
 \begin{figure}[tb]
    \centering
    \includegraphics[width=15.8cm]{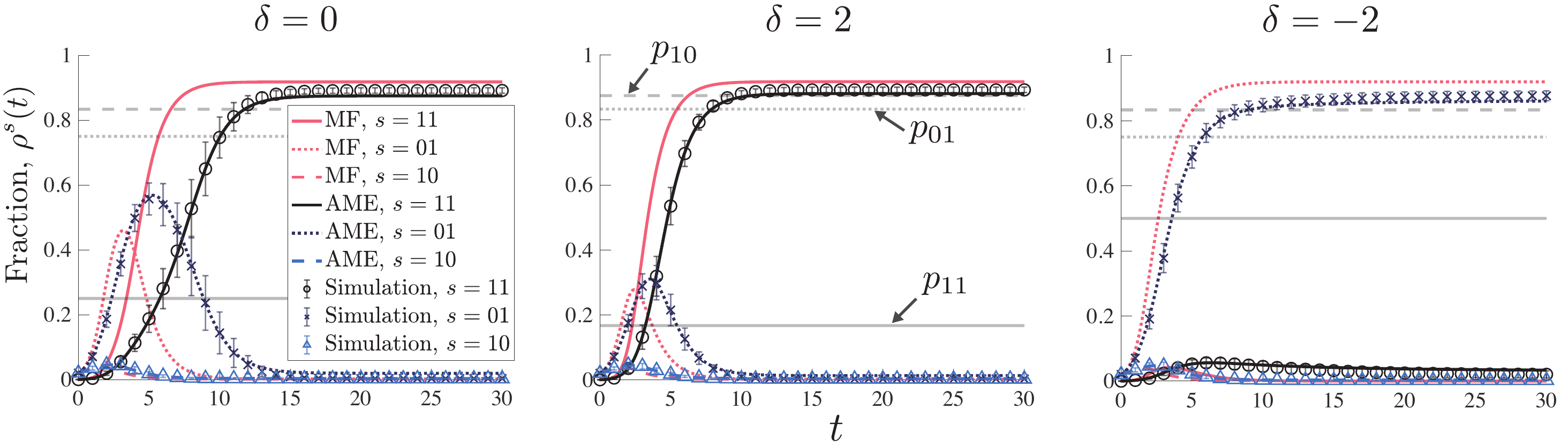}
    \caption{Diffusion of asymmetric activities ($a=6$ and $b=4$). See caption of Fig.~\ref{fig:path_vs_deltas} for details.}
    \label{fig:asymmetric_path_vs_deltas}
\end{figure}
 
 Fig.~\ref{fig:asymmetric_path_vs_deltas} shows the dynamical path of $\rho^s(t)$ when $a>b$ (The other parameters are the same as those in Fig.~\ref{fig:path_vs_deltas}). When the activities are neutral ($\delta=0$, Fig.~\ref{fig:asymmetric_path_vs_deltas}, \emph{left}), strategy~$01$ initially spreads to more than half of the players, and then most of the players begin to shift their strategies to $s=11$, which leads to a decay in the fraction of $01$-players.  In the context of the phase diagram in Fig.~\ref{fig:symmetry_breaking_delta0}, this corresponds to a situation where the feasible region $\rho^{01}+\rho^{10}+\rho^{11}\leq 1$ continues to shrink, and thereby $\rho^{01}$ is bounded from above by $1-\rho^{10}-\rho^{11}$. Note also that $\rho^{01}$ did not exceed $p_{01}$, meaning that $(01,01)$ was not $p$-dominant. 
 
 When the activities are substitutes ($\delta<0$), strategy~$01$ dominates the other strategies (Fig.~\ref{fig:asymmetric_path_vs_deltas}, \emph{right}). Since the total payoff of choosing activity~$B$ cannot be comparable to that of activity~$A$ in the process of diffusion, the simulated paths are quite stable over different runs, and thus the standard deviations are fairly small. Note that $\rho^{01}$ surpasses $p_{01}$ during the diffusion process, which implies that the strategy pair $(01,01)$ is $p$-dominant.

\section{Alternative model of diffusion}\label{sec:utility-based}

We have seen that players' interactions through $4\times 4$ coordination games provide a mechanism by which an activity spreads over a network. To check the robustness of the results, here we consider an alternative model of diffusion in which players select their strategies to maximize utility functions. 

Let $x_i^\ell\in \{0,1\}$ denote player-$i$'s binary action on activity $\ell\in\{A,B\}$. The strategy of player $i$ is expressed by vector ${\bf{x}}_i=(x_i^A,x_i^B)$, and the strategies of the other players are summarized as ${\bf{x}}_{-i}=\allowbreak (x_1^A,x_1^B,\ldots,x_{i-1}^A,x_{i-1}^B, x_{i+1}^A,x_{i+1}^B,\ldots,x_{N}^A,x_{N}^B)$.  Note that there is a one-to-one correspondence between ${\bf{x}}_i\in\{(0,0),(0,1),(1,0),(1,1)\}$ and $s_i\in\{00,01,10,11\}$.  Following \cite{chen2018AEJmultiple}, we consider the following quadratic utility function:
\begin{align}
 u_i(s_i,s_{-i}) = \sum_{\ell=A,B} \left\{ \alpha^\ell x_i^\ell  -\frac{1}{2} (x_i^\ell)^2    + \gamma\sum_{j \neq i}g_{ij}x_i^\ell x_j^\ell \right\} + \frac{1}{2}\beta x_i^A x_i^B
\end{align}
where $g_{ij}\in \{0,1\}$ denotes the $(i,j)$th element of the adjacency matrix $\mathcal{G}=\left( g_{ij}\right)$; $g_{ij}= 1$ if there is an edge between $i$ and $j$ and $g_{ij}=0$ otherwise. $\alpha^\ell$ is a parameter that captures the benefit of enjoying an activity itself, which entails costs ($=1/2$). $\gamma$ and $\beta$ respectively denote the benefit of cooperating with neighbors and the degree of substitutability. Thus, the larger $\beta$ is, the greater the complementarity between the two activities. 

\subsection{Absolute-threshold rules}

The utility for each strategy is given by
\begin{align}
    u_i(00,s_{-i}) &= 0, \\
    u_i(01,s_{-i}) &= \alpha^A -\frac{1}{2} + \gamma (m_{01}+m_{11}), \\
    u_i(10,s_{-i}) &= \alpha^B -\frac{1}{2} + \gamma (m_{10}+m_{11}), \\
    u_i(11,s_{-i}) &= \alpha^A + \alpha^B -1+\frac{\beta}{2} + \gamma (m_{10}+m_{01} + 2m_{11}).
\end{align}
Since the strategy of a player can be considered as a function of the neighbors' strategy profile $\vecm =(m_{00},m_{01},m_{10},m_{11})^\top$, we can redefine the utility function as $\tilde{u}_i(s,\vecm)\equiv u(s,s_{-i})$. The optimal strategy is then given by
\begin{align}
    s^*(\vecm ) =  \underset{s\in S}{\rm arg\,max}\; \tilde{u}(s,\vecm),
    \label{eq:optimal_s_utility}
\end{align}
 where we drop the subscript $i$ for simplicity.
 
 To analyze contagious effects, we impose the following assumptions:
\begin{assumption}
$\alpha^A \leq 1/2$, $\alpha^B \leq 1/2$, and $\alpha^A + \alpha^B -1 + \beta/2 \leq 0$.
\label{ass:utility_parameters}
\end{assumption}
These assumptions guarantee that $\tilde{u}(s,(k,0,0,0)^\top )\leq 0\; \forall s\in S$, which prohibits players from adopting strategies other than $00$ when no neighbors enjoy activities. This means that neighbors' influence is necessary for an activity to spread over the network.
Since $\tilde{u}(00,\vecm)=0$, a player selects a strategy other than $00$ as long as the utility is positive.\footnote{If there are tie values (i.e., $\tilde{u}(s,\vecm)=\tilde{u}(s',\vecm)$ for $s\neq s'$), we randomly select one strategy. If $\tilde{u}(s,\vecm)\leq 0$ for all $s\in S$, however, we select $s=00$.} 

The conditions for $\tilde{u}({01},\vecm)>0$ and $\tilde{u}({10},\vecm)>0$ are respectively given by the following threshold conditions:
\begin{align}
   \tilde{u}({01},{\bf{m}}) > 0 \hspace{8pt} {\rm{iff} }\hspace{8pt} m_{01}+m_{11} > \frac{1}{\gamma}\left(\frac{1}{2}-\alpha^{A}\right), \label{eq:u01_threshold} \\ 
   \tilde{u}({10},{\bf{m}}) > 0 \hspace{8pt} {\rm{iff} }\hspace{8pt} m_{10}+m_{11} > \frac{1}{\gamma}\left(\frac{1}{2}-\alpha^{B}\right).\label{eq:u10_threshold}
\end{align}
The most important difference from the threshold conditions in coordination games, Eqs.~\eqref{eq:v01_threshold} and \eqref{eq:v10_threshold}, is that conditions \eqref{eq:u01_threshold} and \eqref{eq:u10_threshold} are independent of the player's degree $k$. This indicates that a player is more likely to join activity~$A$ as more neighbors join activity~$A$, regardless of the fraction of active neighbors. This also implies that the larger the degree, the more likely it is for an activity to spread over a network. This type of diffusion has been studied within a class of \emph{absolute-threshold models}~\citep{Granovetter1978,Karimi2013PhysicaA,unicomb2021dynamics}.

The conditions $\tilde{u}(11,\vecm)>\tilde{u}(01,\vecm)$ and $\tilde{u}(11,\vecm)>\tilde{u}(10,\vecm)$ are respectively rewritten as
\begin{align}
      \tilde{u}({11},{\bf{m}}) > \tilde{u}(01,\vecm) \hspace{8pt} {\rm{iff} }\hspace{8pt} m_{10}+m_{11} > \frac{1}{\gamma}\left(\frac{1}{2}-\alpha^{B}-\frac{\beta}{2} \right), \label{eq:u01_u11} \\
           \tilde{u}({11},{\bf{m}}) > \tilde{u}(10,\vecm) \hspace{8pt} {\rm{iff} }\hspace{8pt} m_{01}+m_{11} > \frac{1}{\gamma}\left(\frac{1}{2}-\alpha^{A}-\frac{\beta}{2} \right),
     \label{eq:u10_u11}
\end{align}
Suppose for the moment that $\beta> 0$. Then Eq.~\eqref{eq:u01_u11} indicates that the threshold of the number of neighbors enjoying activity~$B$ ($=m_{10}+m_{11}$) above which strategy~$11$ is preferable to strategy~$01$ is smaller than that required for strategy~$10$ to be more preferable than $00$ (Eq.~\ref{eq:u10_threshold}). This is because when $\beta> 0$, the utility of enjoying both activities is greater than or equal to the sum of the utilities for each activity.  
However, if the two activities are substitutable (i.e., $\beta <0$), the incentive to enjoy both activities is discouraged, so the condition for enjoying activity~$B$ \emph{in addition to} $A$ becomes more stringent than condition~\eqref{eq:u10_threshold}, which is the condition to enjoy activity $B$ or do nothing. 
The same argument also holds for the relationship between $\tilde{u}(11,\vecm)$ and $\tilde{u}(10,\vecm)$ (Eq.~\ref{eq:u10_u11}).

 Regarding the choice \emph{between} the two activities, we have 
 \begin{align}
      \tilde{u}(01,\vecm)>\tilde{u}(10,\vecm)\hspace{8pt} {\rm{iff} }\hspace{8pt}
      \alpha^{A}-\alpha^{B} + \gamma(m_{01}-m_{10}) > 0.
      \label{eq:u01_u10}
 \end{align}
 Whether this inequality holds depends on the relative number of active neighbors, $m_{01}-m_{10}$, which will be time-varying.
Thus, a player's strategy may switch from $01$ to $10$ or vice versa in the process of diffusion, depending on the neighbors' status.
 Therefore, diffusion processes in this utility-based model are generally non-monotonic, similar to the contagion processes through coordination games. 

\subsection{Symmetry breaking } 
\begin{figure}[tb]
    \centering
    \includegraphics[width=15.5cm]{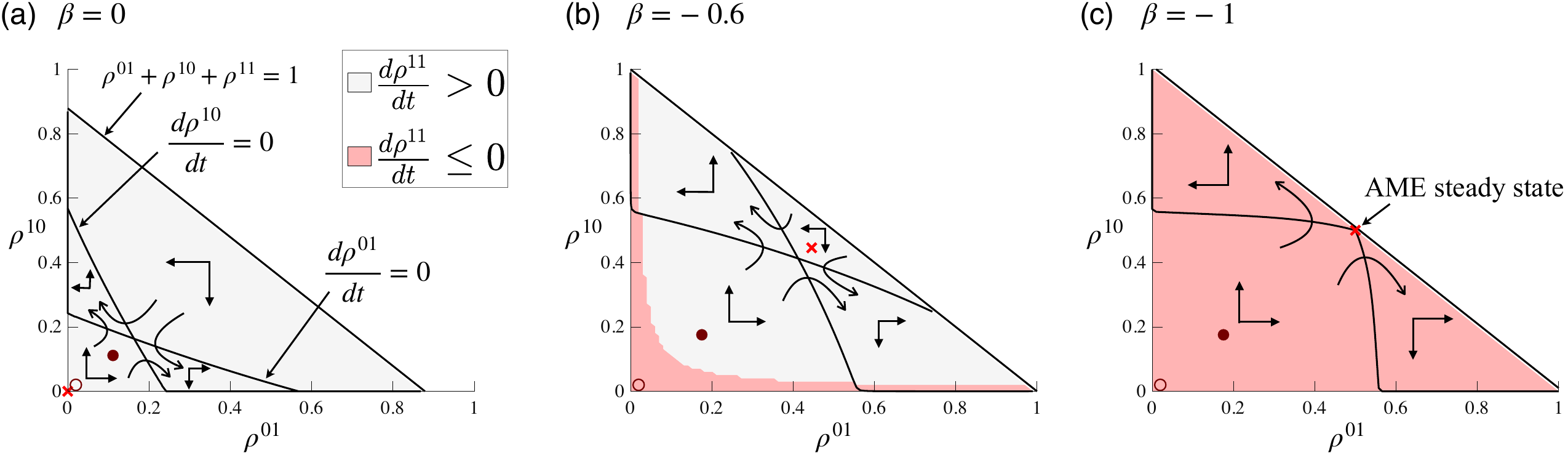}
    \caption{Phase diagram of diffusion through utility-based games. Each panel shows a diagram at $t=1$ based on random 4-regular networks with $\rho_0^{01}=\rho_0^{10}=0.02$, $\rho_0^{11}=0$, $\alpha^A=\alpha^B=0.4$, and $\gamma =0.2$. Red open and closed circles denote the initial and current points, respectively. Red cross denotes the steady state calculated by the AME.}
    \label{fig:symmetry_breaking_utility}
\end{figure}

The previous MF and AME approaches can still be applied in this model by replacing the optimal strategy in the response function \eqref{eq:response_func} with Eq.~\eqref{eq:optimal_s_utility}, which allows us to draw phase diagrams in the same way as before based on random regular networks (Fig.~\ref{fig:symmetry_breaking_utility}). Under our benchmark parameters $N=3000$, $\rho_0^{01}=\rho_0^{10}=0.02$, $\rho_0^{11}=0$, $\alpha^A=\alpha^B=0.4$, and $\gamma =0.2$, we find that $d\rho^{11}/dt$ is always positive when the two activities are neutral (Fig.~\ref{fig:symmetry_breaking_utility}a), which causes the feasible region of $(\rho^{01},\rho^{10})$ to shrink over time. Therefore, we will have $(\rho^{01}(T),\rho^{10}(T),\rho^{11}(T)) = (0,0,1)$ at the AME steady state when $\beta=0$ (denoted by red cross). 

On the other hand, as the substitutability parameter $\beta$ becomes negative, there arises a wider region within which $d\rho^{11}/dt \leq 0$ (Figs.~\ref{fig:symmetry_breaking_utility}b and c). In particular, when $\beta = -1$, we see that $d\rho^{11}/dt$ is non-positive for the entire feasible region of $(\rho^{01},\rho^{10})$. This suggests that the steady state obtained by the AME method is practically not attainable since any path deviating from the stable path (i.e., stable arm) would cause symmetry breaking, leading to $(\rho^{01}(T),\rho^{10}(T),\rho^{11}(T)) = (1,0,0)$ or $(0,1,0)$. Again, this is a situation in which one of the activities would dominate the other purely by chance, depending on the realization of an instance drawn from an ensemble of random networks. When the degree of substitutability is moderate (Fig.~\ref{fig:symmetry_breaking_utility}b), there is a wide region of $(\rho^{01},\rho^{10})$ in which $d\rho^{11}/dt>0$, so the feasible region may shrink over time. This will make the upper bounds for $\rho^{01}$ and $\rho^{10}$ well below $1$, resulting in a smaller degree of asymmetry compared to the case of $\beta=-1$.  

Fig.~\ref{fig:vsdelta_utility}a illustrates how the accuracy of the AME solutions is related to the degree of substitutability $\beta$ and the average degree $z$. We find that when $\beta=-1$, each simulation run will be likely to deviate from the AME solution, while the average over multiple simulations appears to be close to the value predicted by the AME. These deviations from the AME values can be interpreted as an outcome of incidental diffusion due to symmetry breaking, whose dynamics are captured by the phase diagrams in Fig.~\ref{fig:symmetry_breaking_utility}c.   

\begin{figure}[tb]
    \centering
    \includegraphics[width=15.5cm]{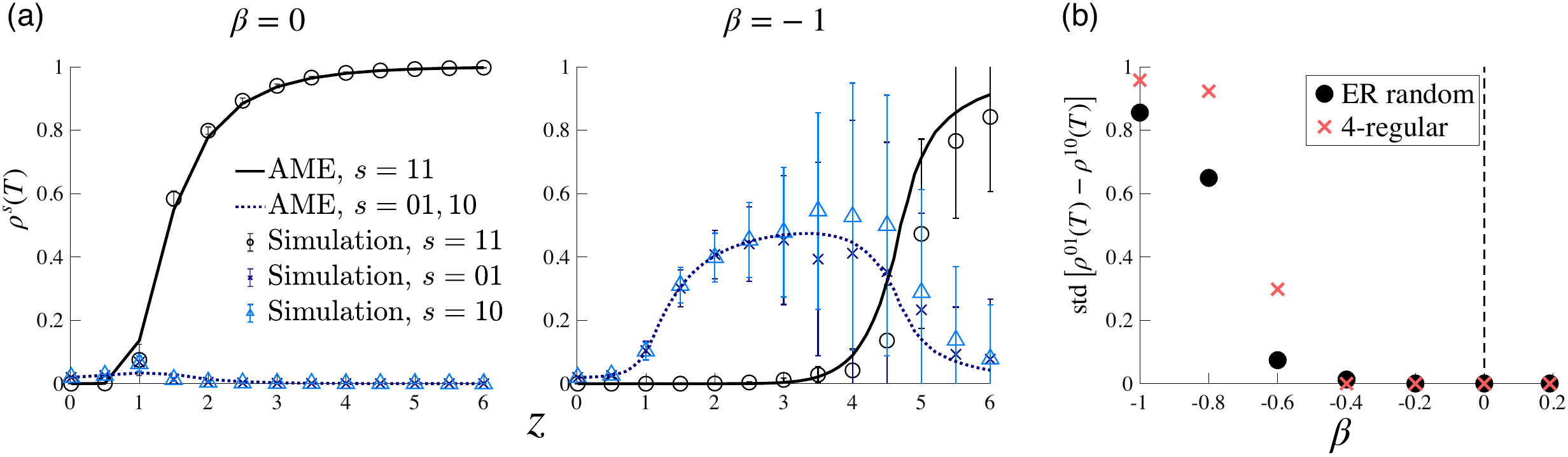}
    \caption{Effect of substitutability on symmetry breaking in the model of utility-based games. (a) Steady-state fraction of strategy $s$ against $z$ for a given substitutability $\beta$ in \ER networks. Error bars denote one standard deviations over $100$ runs. (b) Relationship between the standard deviations of $\rho^{01}(T)-\rho^{10}(T)$ and $\beta$ for $z=4$, $N=3000$, $\rho_0^{01}=\rho_0^{10}=0.02$, $\rho_0^{11}=0$, $\alpha^A=\alpha^B=0.4$, $\gamma =0.2$. $T$ is set at $50$ and $100$ for panels (a) and (b), respectively.  }
    \label{fig:vsdelta_utility}
\end{figure}

Naturally, we find that there is a negative relationship between the degree of substitutability $\beta$ and the standard deviations of $\rho^{01}(T)-\rho^{10}(T)$ (Fig.~\ref{fig:vsdelta_utility}b). When the two activities are strongly substitutable, it is highly likely that either of the two dominates the other (i.e., $\rho^{01}(T)=1$ or $\rho^{10}(T)=1$), resulting in the standard deviations being close to $1$. When the substitutability is moderate, symmetry breaking still occurs, but the difference between $\rho^{01}(T)$ and $\rho^{10}(T)$ would be well below $1$ because strategy~$11$ also prevails to some extent. In the phase diagram in Fig.~\ref{fig:symmetry_breaking_utility}b, this corresponds to an expansion of the area in which $d\rho^{11}/dt > 0$. 

\subsection{Asymmetric preferences}\label{sec:asymmetric_utility}

 Now we analyze the case of asymmetric preferences where there is no possibility of symmetry breaking by definition. In particular, we investigate the relationship between network connectivity and the steady-state share of each strategy. Throughout this section, we assume that activity $B$ is a little less attractive than activity $A$, where $\alpha^{A}= 0.4$ and $\alpha^{B}=0.3$.
 
 When the two activities are neutral (i.e., $\beta=0$), we see that lower connectivity would lead to a higher chance of strategy~$01$ being propagated over a network (Fig.~\ref{fig:asymmetry_utility}, \emph{left}). Under our parameter configurations, strategy~$01$ is the dominant strategy up to $z\approx 2.5$, but for more densely connected networks, strategy~$11$ dominates and both activities $A$ and $B$ are likely to prevail. Thus, the dominant strategy switches around $z\approx 2.5$, at which the simulated cascade sizes can be highly volatile and their average would not match the AME solution unless the number of runs is large enough.\footnote{In a simple binary cascade model, \cite{Watts2002} shows that cascade sizes are power-law distributed at a critical point above which large-size cascades can occur.}  
 When the activities are complements (e.g., $\beta=0.1$), strategy~$11$ would be dominant in well connected networks such that $z>1$ (Fig.~\ref{fig:asymmetry_utility}, \emph{middle}). Indeed, both activities will spread to most of the players located in the giant component of a network, as is seen in the standard binary cascade models~\citep{Watts2002}. When the activities are strongly substitutable (e.g., $\beta=-1$), only activity $A$ would gain popularity, while there would be no chance for $B$ to propagate (Fig.~\ref{fig:asymmetry_utility}, \emph{right}).

 It should be noted that there are some differences between the current utility-based model and the model of coordination games concerning the role of network connectivity. 
 First, in the diffusion mechanism through coordination games, there is an upper bound of $z$ above which propagation decays (Fig.~\ref{fig:heatmap}). This is because players' thresholds are given by the fractions of neighbors already participating in certain activities, where an increase in the number of neighbors would dilute the influence of each neighbor. However, such a dilution effect does not exist in the diffusion process through utility-based games since the threshold conditions \eqref{eq:u01_threshold}--\eqref{eq:u01_u10} depend only on the absolute numbers of neighbors adopting certain strategies. Therefore, there is no upper bound of $z$ for the activities to propagate in the utility-based model. 
 
 \begin{figure}[t]
    \centering
    \includegraphics[width=15.8cm]{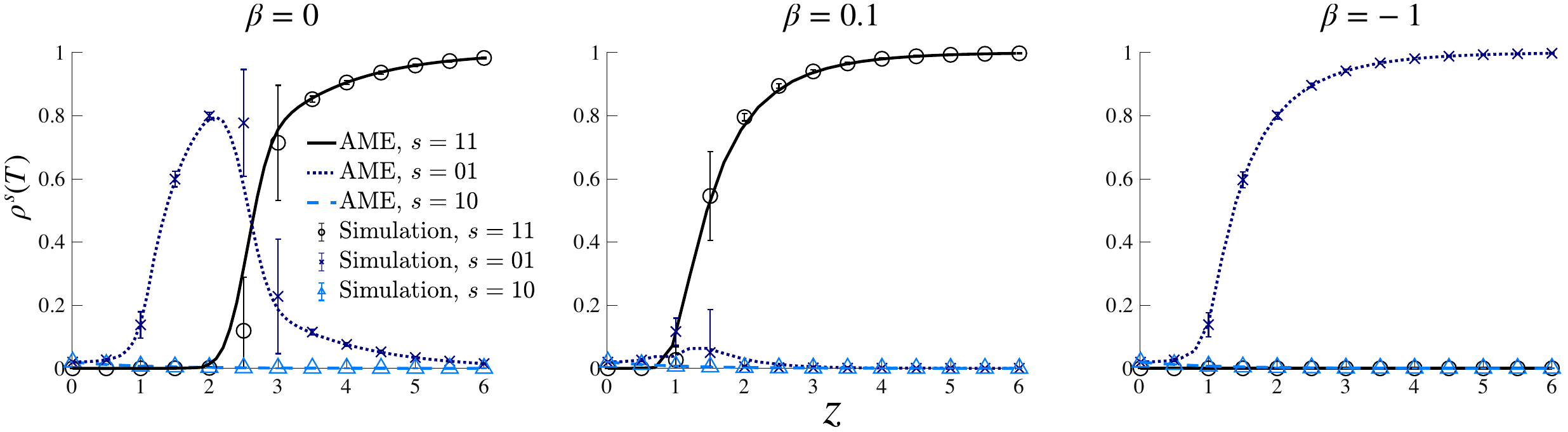}
    \caption{Steady state of diffusion through utility-based games with asymmetric preferences ($\alpha^A=0.4$ and $\alpha^B=0.3$). See caption of Fig.~\ref{fig:vsdelta_utility} for details.
    }
    \label{fig:asymmetry_utility}
\end{figure}

 Second, while there is a common property that a larger degree of substitutability will lead players to select a more attractive activity, we see an essential difference when the two activities are neutral. In the diffusion through utility-based games, lower connectivity would lead to a propagation of a single activity, as shown in Fig.~\ref{fig:asymmetry_utility}. In contrast, in the diffusion through coordination games, a lower degree of connectivity will enhance the propagation of both activities, while a higher degree of connectivity would allow either activity to spread (Fig.~\ref{fig:heatmap}). In the fractional threshold model, the influence of a single neighbor will be diluted as the number of neighbors increases, and accordingly, it becomes harder for an activity to gain popularity. In contrast, in the absolute threshold model, a rise in the number of neighbors just facilitates the transmission of peer effects, as it increases the number of routes through which information arrives.

\section{Empirical social networks}

Now we examine how accurately the AME approach can describe the dynamics of diffusion in a real-world social network. For this purpose, we construct a social network of economists based on the acknowledgments of the papers published in \textit{American Economic Review} (AER).  

\subsection{Data: a social network of AER authors}
To construct a network, we first create a list of authors who published at least one paper in AER between January 2019 and December 2020. During this period AER published 244 papers. 
Two authors are connected by an undirected and unweighted edge if one of them gives credit to the other in the acknowledgments of a published paper, assuming that there is a social tie between them.\footnote{To extract author names from the PDF file of each paper that was downloaded manually from the AER web site, we used a Python module Stanza~\citep{stanza}} The constructed network consists of some small isolated components, so we focus on the largest connected component on which a global diffusion could occur. A visualization of the network structure is presented in Fig.~\ref{fig:AER_graph}. We have $N=447$ with mean degree $z=4.99$ and maximum degree $k_{\rm max}=29$. 
\begin{figure}[tb]
    \centering
    \includegraphics[width=8cm]{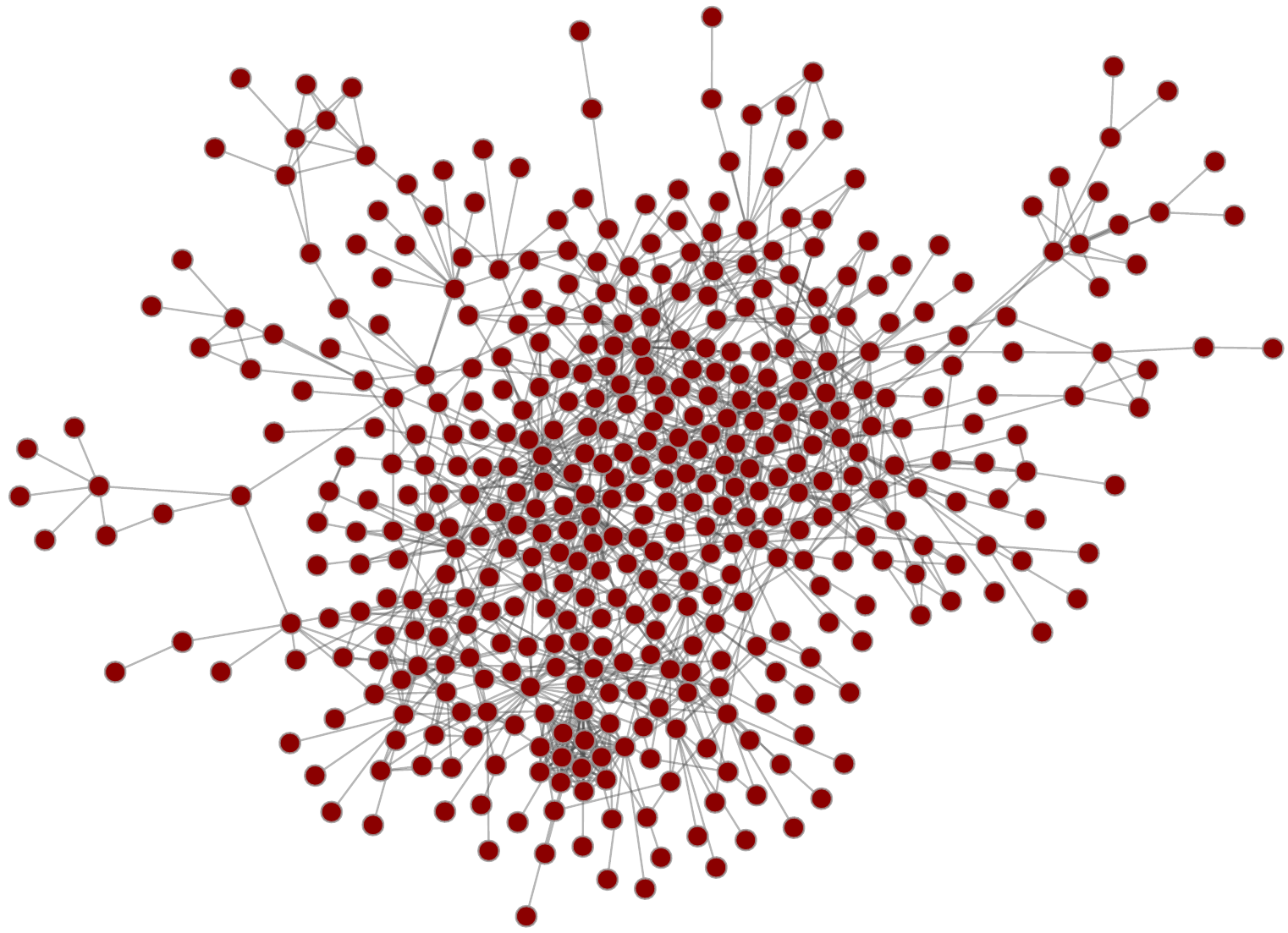}
    \caption{Social network of AER authors, 2019--2020. Nodes and edges respectively represent authors and the presence of mentions in the acknowledgments ($N=447$ and $z=4.99$).} 
    \label{fig:AER_graph}
\end{figure}
\begin{figure}[tb]
    \centering
    \includegraphics[width=15.8cm]{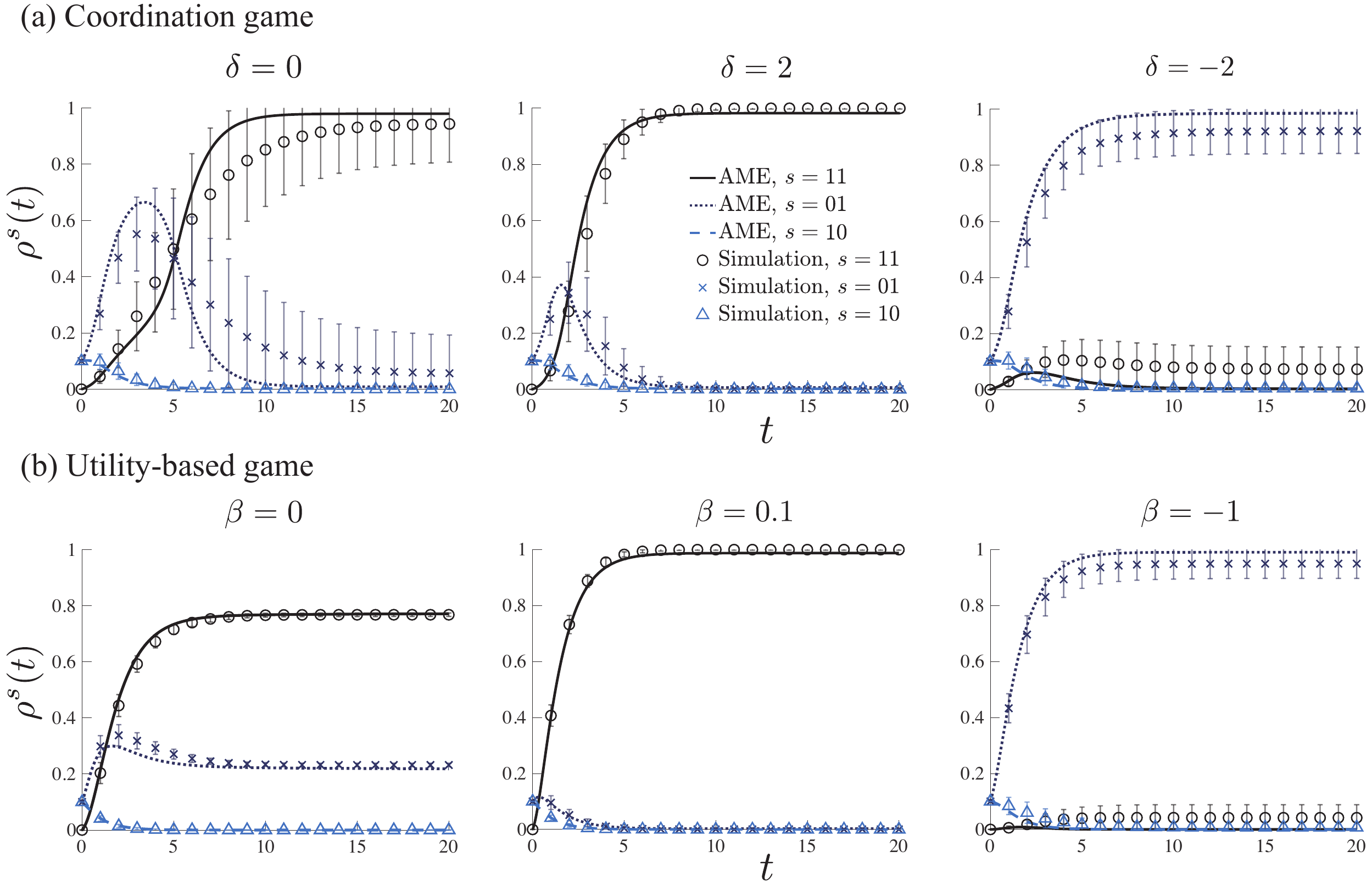}
    \caption{Diffusion processes on the AER author network. (a) Coordination game: $a=6$, $b=4$, and $c=1$. (b) Utility-based game: $\alpha^A=0.4$, $\alpha^B=0.3$, $\gamma =0.2$. In both panels we set $\rho_0^{01}=\rho_0^{10}=0.1$, and $\rho_0^{11}=0$. Error bars denote one standard deviations over $1000$ runs.}
    \label{fig:AER_sim}
\end{figure}

\subsection{Diffusion on a social network of economists}

We examine how well the AME will predict the dynamical paths for each strategy. We consider asymmetric activities in the same way as in sections~\ref{sec:asymmetric_coordination} and \ref{sec:asymmetric_utility} (i.e., $a>b$ and $\alpha >\beta$) since the AME works well on synthetic networks when there is intrinsic asymmetry.
To calculate the AME solution, we obtain the empirical degree distribution $p_k^{\rm e}$ as
\begin{align}
    p_{k}^{\rm e} = \frac{\# \text{ of nodes with degree $k$}}{N}.
\end{align}
We find that in both types of games, the simulated diffusion dynamics are well replicated by the AME approach (Fig.~\ref{fig:AER_sim}). A difference from the simulations based on synthetic networks is that we see larger standard deviations for the fraction of each strategy, indicated by the error bars, while the two activities are intrinsically asymmetric. This would be because the size of the author network is relatively small, thereby deteriorating the accuracy of the approximation method which is expected to work well for sufficiently large networks. To mitigate the small-size problem, we imposed a relatively high value for the seed fractions,  $\rho_0^{01}=\rho_0^{10}=0.1$, yet we still see some fluctuations in the simulated results. Arguably, the standard deviations would diminish if we use larger empirical networks, such as networks aggregated over longer periods. However, the maximum degree $k_{\rm max}$ of such a large network tends to be so high that the computational cost of the AME would be prohibitively expensive.

 \section{Conclusion and discussion}
 
 We studied the dynamics of diffusion based on two different models of network games. In both models, we showed that the diffusion dynamics exhibit instability due to phase transitions and symmetry breaking. When the activities are complements or neutral, the AME approach is highly accurate in predicting the popularity of each activity at a given point in time. On the other hand, when the activities are substitutable, ``average-based'' analytical equations such as MF and AME may not correctly describe the Nash equilibria attained in simulated propagation processes.
 
 There are some issues to be addressed in future research. 
 First, while we considered two activities and four pure strategies $\{00,01,10,11\}$, one could extend the models to study more than two activities since, in principle, the AME method can be applied for any number of states (or strategies) $n$. However, the difficulty is that the number of differential equations will increase to $n\sum_{k=0}^{k_{\rm max}}\binom{n+k-1}{k}$, so the computational cost would easily become prohibitive. For instance, if we have three types of activities and a degree distribution over $k=0$ to $10$, then the number of possible strategies is given by $n=2^3$, and the number of differential equations will be $8\sum_{k=0}^{10}\binom{8+k-1}{k}= 350,064$. 
 Second, it would be useful to introduce non-random structural properties, such as clustering and assortativity, which are frequently observed in real social networks~\citep{Barabasi2016book,newman2018book2nd}. 
 Since the AME system is already complex and there are a number of equations to be solved, introducing additional network properties would be a difficult task. Nevertheless, incorporating structural properties would be important not only to improve the accuracy of the AME running on 
 real data, but also enhances the predictability of diffusion. For example, firms would aim to enhance the popularity of new products through online social networks~\citep{Watts2007}, and governments may need to monitor the spread of misinformation and manipulations that could spread through social media~\citep{Badaway2019Political}. Third, it would also be useful to take into account the time-varying aspects of real-world networks. In many social networks, nodes and edges frequently appear and disappear, leading to changes in the network structure.
 The dynamic nature of networks has been extensively investigated within the framework of ``\emph{temporal networks}'' in the field of network science~\citep{HolmeSaramaki2013book_Springer}. We hope that our study will stimulate further research in these directions.



\begin{thebibliography}{}

\bibitem[\protect\citeauthoryear{Acemoglu, Robinson, and Verdier}{Acemoglu
  et~al.}{2017}]{acemoglu2017asymmetric}
Acemoglu, D., J.~A. Robinson, and T.~Verdier (2017).
\newblock Asymmetric growth and institutions in an interdependent world.
\newblock {\em Journal of Political Economy\/}~{\em 125\/}, 1245--1305.

\bibitem[\protect\citeauthoryear{Arigapudi}{Arigapudi}{2020}]{arigapudi2020bilingual}
Arigapudi, S. (2020).
\newblock Transitions between equilibria in bilingual games under logit choice.
\newblock {\em Journal of Mathematical Economics\/}~{\em 86}, 24--34.

\bibitem[\protect\citeauthoryear{Badawy, Lerman, and Ferrara}{Badawy
  et~al.}{2019}]{Badaway2019Political}
Badawy, A., K.~Lerman, and E.~Ferrara (2019).
\newblock Who falls for online political manipulation?
\newblock In {\em Companion Proceedings of The 2019 World Wide Web Conference}, 162–168.

\bibitem[\protect\citeauthoryear{Ballester, Calv{\'o}-Armengol, and
  Zenou}{Ballester et~al.}{2006}]{ballester2006}
Ballester, C., A.~Calv{\'o}-Armengol, and Y.~Zenou (2006).
\newblock Who's who in networks. Wanted: The key player.
\newblock {\em Econometrica\/}~{\em 74\/}, 1403--1417.

\bibitem[\protect\citeauthoryear{Barab{\'{a}}si}{Barab{\'{a}}si}{2016}]{Barabasi2016book}
Barab{\'{a}}si, A.-L. (2016).
\newblock {\em {Network Science}}.
\newblock Cambridge University Press, Cambridge.

\bibitem[\protect\citeauthoryear{Bonacich}{Bonacich}{1987}]{bonacich1987power}
Bonacich, P. (1987).
\newblock Power and centrality: A family of measures.
\newblock {\em American Journal of Sociology\/}~{\em 92\/}, 1170--1182.

\bibitem[\protect\citeauthoryear{Chatterjee}{Chatterjee}{2017}]{chatterjee2017endogenous}
Chatterjee, A. (2017).
\newblock Endogenous comparative advantage, gains from trade and
  symmetry-breaking.
\newblock {\em Journal of International Economics\/}~{\em 109}, 102--115.

\bibitem[\protect\citeauthoryear{Chen, Zenou, and Zhou}{Chen
  et~al.}{2018}]{chen2018AEJmultiple}
Chen, Y.-J., Y.~Zenou, and J.~Zhou (2018).
\newblock Multiple activities in networks.
\newblock {\em American Economic Journal: Microeconomics\/}~{\em 10\/}, 34--85.

\bibitem[\protect\citeauthoryear{Erd\H{o}s and R\'enyi}{Erd\H{o}s and
  R\'enyi}{1959}]{Erdos1959PublMath}
Erd\H{o}s, P. and A.~R\'enyi (1959).
\newblock On random graphs.
\newblock {\em Publicationes Mathematicae\/}~{\em 6}, 290--297.

\bibitem[\protect\citeauthoryear{Fennell and Gleeson}{Fennell and
  Gleeson}{2019}]{fennell2019multistate}
Fennell, P.~G. and J.~P. Gleeson (2019).
\newblock Multistate dynamical processes on networks: Analysis through
  degree-based approximation frameworks.
\newblock {\em SIAM Review\/}~{\em 61\/}, 92--118.

\bibitem[\protect\citeauthoryear{Gai and Kapadia}{Gai and
  Kapadia}{2010}]{GaiKapadia2010}
Gai, P. and S.~Kapadia (2010).
\newblock {Contagion in financial networks}.
\newblock {\em Proceedings of the Royal Society A\/}~{\em 466}, 2401--2423.

\bibitem[\protect\citeauthoryear{Gleeson}{Gleeson}{2011}]{gleeson2011high}
Gleeson, J.~P. (2011).
\newblock High-accuracy approximation of binary-state dynamics on networks.
\newblock {\em Physical Review Letters\/}~{\em 107\/}, 068701.

\bibitem[\protect\citeauthoryear{Gleeson}{Gleeson}{2013}]{gleeson2013binary}
Gleeson, J.~P. (2013).
\newblock Binary-state dynamics on complex networks: Pair approximation and
  beyond.
\newblock {\em Physical Review X\/}~{\em 3\/}, 021004.

\bibitem[\protect\citeauthoryear{Gleeson and Porter}{Gleeson and
  Porter}{2018}]{gleeson2018message}
Gleeson, J.~P. and M.~A. Porter (2018).
\newblock Message-passing methods for complex contagions.
\newblock In S. Lehmann and Y.-Y. Ahn (Eds.), {\em Complex Spreading Phenomena in Social Systems}, 81--95. Springer.

\bibitem[\protect\citeauthoryear{Goyal and Janssen}{Goyal and
  Janssen}{1997}]{goyal1997non}
Goyal, S. and M.~C. Janssen (1997).
\newblock Non-exclusive conventions and social coordination.
\newblock {\em Journal of Economic Eheory\/}~{\em 77\/}, 34--57.

\bibitem[\protect\citeauthoryear{Granovetter}{Granovetter}{1978}]{Granovetter1978}
Granovetter, M. (1978).
\newblock {Threshold models of collective behavior}.
\newblock {\em American Journal of Sociology\/}~{\em 83\/}, 1420--1443.

\bibitem[\protect\citeauthoryear{Harsanyi and Selten}{Harsanyi and Selten}{1988}]{harsanyi1988general}
Harsanyi, J.~C., R.~Selten (1988).
\newblock {\em A General Theory of Equilibrium Selection in Games.}
\newblock MIT Press, Cambridge. 

\bibitem[\protect\citeauthoryear{Holme and Saram\"{a}ki}{Holme and
  Saram\"{a}ki}{2013}]{HolmeSaramaki2013book_Springer}
Holme, P. and J.~Saram\"{a}ki (2013).
\newblock {\em {Temporal Networks}}.
\newblock Springer-Verlag, Berlin.

\bibitem[\protect\citeauthoryear{Immorlica, Kleinberg, Mahdian, and
  Wexler}{Immorlica et~al.}{2007}]{immorlica2007role}
Immorlica, N., J.~Kleinberg, M.~Mahdian, and T.~Wexler (2007).
\newblock The role of compatibility in the diffusion of technologies through
  social networks.
\newblock In {\em Proceedings of the 8th ACM Conference on Electronic
  Commerce}, 75--83.

\bibitem[\protect\citeauthoryear{Jackson}{Jackson}{2008}]{Jackson2008book}
Jackson, M.~O. (2008).
\newblock {\em Social and Economic Networks}.
\newblock Princeton University Press, Princeton.

\bibitem[\protect\citeauthoryear{Jackson and Yariv}{Jackson and
  Yariv}{2007}]{jackson2007diffusion}
Jackson, M.~O. and L.~Yariv (2007).
\newblock Diffusion of behavior and equilibrium properties in network games.
\newblock {\em American Economic Review\/}~{\em 97\/}, 92--98.

\bibitem[\protect\citeauthoryear{Jackson and Zenou}{Jackson and
  Zenou}{2015}]{jackson-zenou2015games}
Jackson, M.~O. and Y.~Zenou (2015).
\newblock Games on networks.
\newblock In {\em Handbook of Game Theory with Economic Applications 4}, 95--163. Elsevier.

\bibitem[\protect\citeauthoryear{Kajii and Morris}{Kajii and
  Morris}{1997}]{kajii1997robustness}
Kajii, A. and S.~Morris (1997).
\newblock The robustness of equilibria to incomplete information.
\newblock {\em Econometrica 65\/}, 1283--1309.

\bibitem[\protect\citeauthoryear{Karimi and Holme}{Karimi and
  Holme}{2013}]{Karimi2013PhysicaA}
Karimi, F. and P.~Holme (2013).
\newblock {Threshold model of cascades in empirical temporal networks}.
\newblock {\em Physica A\/}~{\em 392}, 3476--3483.

\bibitem[\protect\citeauthoryear{Kobayashi and Onaga}{Kobayashi and
  Onaga}{2021}]{kobayashi2021dynamics}
Kobayashi, T. and T.~Onaga (2021).
\newblock Dynamics of diffusion on monoplex and multiplex networks: A message-passing approach.
\newblock {\em SSRN 3806211\/}.

\bibitem[\protect\citeauthoryear{Kreindler and Young}{Kreindler and
  Young}{2014}]{kreindler2014rapid}
Kreindler, G.~E. and H.~P. Young (2014).
\newblock Rapid innovation diffusion in social networks.
\newblock {\em Proceedings of the National Academy of Sciences USA\/}~{\em
  111\/}, 10881--10888.

\bibitem[\protect\citeauthoryear{Lelarge}{Lelarge}{2012}]{lelarge2012diffusion}
Lelarge, M. (2012).
\newblock Diffusion and cascading behavior in random networks.
\newblock {\em Games and Economic Behavior\/}~{\em 75\/}, 752--775.

\bibitem[\protect\citeauthoryear{L{\'o}pez-Pintado}{L{\'o}pez-Pintado}{2008}]{lopez2008GEBdiffusion}
L{\'o}pez-Pintado, D. (2008).
\newblock Diffusion in complex social networks.
\newblock {\em Games and Economic Behavior\/}~{\em 62\/}, 573--590.

\bibitem[\protect\citeauthoryear{L{\'o}pez-Pintado}{L{\'o}pez-Pintado}{2012}]{lopez2012GEBinfluence}
L{\'o}pez-Pintado, D. (2012).
\newblock Influence networks.
\newblock {\em Games and Economic Behavior\/}~{\em 75\/}, 776--787.

\bibitem[\protect\citeauthoryear{Matsuyama}{Matsuyama}{2002}]{matsuyama2002explaining}
Matsuyama, K. (2002).
\newblock Explaining diversity: Symmetry-breaking in complementarity games.
\newblock {\em American Economic Review\/}~{\em 92\/}, 241--246.

\bibitem[\protect\citeauthoryear{Matsuyama}{Matsuyama}{2004}]{matsuyama2004financial}
Matsuyama, K. (2004).
\newblock Financial market globalization, symmetry-breaking, and endogenous
  inequality of nations.
\newblock {\em Econometrica\/}~{\em 72\/}, 853--884.

\bibitem[\protect\citeauthoryear{Matsuyama}{Matsuyama}{2013}]{matsuyama2013trade}
Matsuyama, K. (2013).
\newblock Endogenous ranking and equilibrium lorenz curve across (ex ante)
  identical countries.
\newblock {\em Econometrica\/}~{\em 81\/}, 2009--2031.

\bibitem[\protect\citeauthoryear{McKay, Wormald, and Wysocka}{McKay
  et~al.}{2004}]{mckay2004short}
McKay, B.~D., N.~C. Wormald, and B.~Wysocka (2004).
\newblock Short cycles in random regular graphs.
\newblock {\em Electronic Journal of Combinatorics\/}~{\em 11}, 66--66.

\bibitem[\protect\citeauthoryear{Molloy and Reed}{Molloy and
  Reed}{1995}]{molloy1995critical}
Molloy, M. and B.~Reed (1995).
\newblock A critical point for random graphs with a given degree sequence.
\newblock {\em Random Structures \& Algorithms\/}~{\em 6\/}, 161--180.

\bibitem[\protect\citeauthoryear{Morris}{Morris}{2000}]{morris2000contagion}
Morris, S. (2000).
\newblock Contagion.
\newblock {\em Review of Economic Studies\/}~{\em 67\/}, 57--78.

\bibitem[\protect\citeauthoryear{Morris, Rob, and Shin}{Morris
  et~al.}{1995}]{morris1995p-dominance}
Morris, S., R.~Rob, and H.~S. Shin (1995).
\newblock $p$-dominance and belief potential.
\newblock {\em Econometrica\/}~{\em 63}, 145--157.

\bibitem[\protect\citeauthoryear{Newman}{Newman}{2018}]{newman2018book2nd}
Newman, M. (2018).
\newblock {\em Networks, 2nd ed.}
\newblock Oxford University Press, Oxford.

\bibitem[\protect\citeauthoryear{Oyama and Takahashi}{Oyama and
  Takahashi}{2015}]{oyama2015bilingual}
Oyama, D. and S.~Takahashi (2015).
\newblock Contagion and uninvadability in local interaction games: The
  bilingual game and general supermodular games.
\newblock {\em Journal of Economic Theory\/}~{\em 157}, 100--127.

\bibitem[\protect\citeauthoryear{Qi, Zhang, Zhang, Bolton, and Manning}{Qi
  et~al.}{2020}]{stanza}
Qi, P., Y.~Zhang, Y.~Zhang, J.~Bolton, and C.~D. Manning (2020).
\newblock {S}tanza: A python natural language processing toolkit for many human
  languages.
\newblock In {\em Proceedings of the 58th Annual Meeting of the Association for
  Computational Linguistics: System Demonstrations}, 101--108.

\bibitem[\protect\citeauthoryear{Sadler}{Sadler}{2020}]{sadler2020diffusion}
Sadler, E. (2020).
\newblock Diffusion games.
\newblock {\em American Economic Review\/}~{\em 110\/}, 225--70.

\bibitem[\protect\citeauthoryear{Unicomb, I{\~n}iguez, Gleeson, and
  Karsai}{Unicomb et~al.}{2021}]{unicomb2021dynamics}
Unicomb, S., G.~I{\~n}iguez, J.~P. Gleeson, and M.~Karsai (2021).
\newblock Dynamics of cascades on burstiness-controlled temporal networks.
\newblock {\em Nature Communications\/}~{\em 12\/}, 1--10.

\bibitem[\protect\citeauthoryear{Watts}{Watts}{2002}]{Watts2002}
Watts, D.~J. (2002).
\newblock A simple model of global cascades on random networks.
\newblock {\em Proceedings of the National Academy of Sciences USA\/}~{\em 99\/}, 5766--5771.

\bibitem[\protect\citeauthoryear{Watts and Dodds}{Watts and
  Dodds}{2007}]{Watts2007}
Watts, D.~J. and P.~S. Dodds (2007).
\newblock {Influentials, networks, and public opinion formation}.
\newblock {\em Journal of Consumer Research\/}~{\em 34\/}, 441--458.

\bibitem[\protect\citeauthoryear{Wormald}{Wormald}{1999}]{wormald1999models}
Wormald, N.~C. (1999).
\newblock Models of random regular graphs.
\newblock In J.~D. Lamb and D.~A. Preece (Eds.), {\em Surveys in Combinatorics}, 239--298. Cambridge University Press, Cambridge.

\bibitem[\protect\citeauthoryear{Young}{Young}{1993}]{young1993evolution}
Young, H.~P. (1993).
\newblock The evolution of conventions.
\newblock {\em Econometrica\/}~{\em 61}, 57--84.

\bibitem[\protect\citeauthoryear{Young}{Young}{2011}]{young2011dynamics}
Young, H.~P. (2011).
\newblock The dynamics of social innovation.
\newblock {\em Proceedings of the National Academy of Sciences USA\/}~{\em
  108\/}, 21285--21291.

\end{thebibliography}

\end{CJK*}

\end{document}